\documentclass[envcountsect,envcountsame,final,orivec]{llncs}

\usepackage{times}
\usepackage{url}
\usepackage[nocompress]{cite}

\usepackage{tikz}
\usetikzlibrary{arrows,decorations.pathmorphing,snakes,backgrounds,positioning,fit}
\usetikzlibrary{automata}

\usepackage[left=3cm,right=3cm,top=3.8cm,bottom=3.8cm]{geometry} \usepackage[latin1]{inputenc}
\usepackage{amsmath}
\usepackage{mathtools}
\usepackage{amssymb} 
\usepackage{stmaryrd}
\usepackage[mathcal,mathscr]{euscript}
\usepackage[english]{babel}
\usepackage{listings}
\usepackage{proof}
\usepackage{graphicx}
\usepackage{enumerate}
\usepackage[citecolor=blue,linkcolor=blue]{hyperref}
\hypersetup{
    colorlinks,
    linkcolor={blue!65!black},
    citecolor={blue!65!black},
    urlcolor={blue!65!black}
}
\usepackage[ruled,noend,noline]{algorithm2e}

\SetCommentSty{mycommfont}

\usepackage{esvect} 
\usepackage{bm}

\spnewtheorem{remarkbf}[theorem]{Remark}{\bfseries}{\rmfamily}
\spnewtheorem{examplebf}[theorem]{Example}{\bfseries}{\rmfamily}
\spnewtheorem{assumption}[theorem]{Assumption}{\bfseries}{\rmfamily}

\def \tuple#1{\langle #1 \rangle}

\newcommand{\Ra}{\Rightarrow}
\newcommand{\La}{\Leftarrow}
\newcommand{\Lra}{\Leftrightarrow}

\newcommand{\ra}{\rightarrow}
\newcommand{\sra}{\shortrightarrow}

\newcommand{\ud}{\triangleq}
\newcommand{\diff}{\stackrel{\text{\tiny $\triangle$}}{\Leftrightarrow}}
\newcommand{\dequiv}{\stackrel{\text{\tiny $\triangle$}}{\equiv}}

\def\ok#1{\mbox{\raisebox{0ex}[1ex][1ex]{$#1$}}}

\newcommand{\umul}{{\check{\mu}_L}}
\newcommand{\lmul}{{\hat{\mu}_L}}

\newcommand{\rat}[1]{\stackrel{\text{\tiny $#1$}}{\rightarrow}}
\newcommand{\grasse}[1]{\llbracket #1 \rrbracket}

\DeclareMathOperator{\uco}{uco}
\DeclareMathOperator{\lco}{lco}
\DeclareMathOperator{\av}{\mathit{Av}}
\DeclareMathOperator{\Var}{Var}

\DeclareMathOperator{\aff}{aff}
\DeclareMathOperator{\AInv}{{\normalfont \textsc{AInv}}}
\DeclareMathOperator{\Aff}{{\normalfont \textsc{Aff}}}
\DeclareMathOperator{\Const}{{\normalfont \textsc{Const}}}

\DeclareMathOperator{\C}{{\normalfont \textsc{C}}}

\DeclareMathOperator{\safe}{{\normalfont \textsc{safe}}}
\DeclareMathOperator{\inv}{{\normalfont \textsc{inv}}}

\DeclareMathOperator{\LExp}{LExp}
\DeclareMathOperator{\Reach}{Reach}

\DeclareMathOperator{\pre}{\mathsf{pre}}
\DeclareMathOperator{\post}{\mathsf{post}}
\DeclareMathOperator{\pret}{\widetilde{\mathsf{pre}}}
\DeclareMathOperator{\postt}{\widetilde{\mathsf{post}}}

\DeclareMathOperator{\lfp}{\mathsf{lfp}}
\DeclareMathOperator{\gfp}{\mathsf{gfp}}
\DeclareMathOperator{\Fix}{Fix}
\DeclareMathOperator{\preFix}{Fix^{\text{\tiny $\leq$}}}
\DeclareMathOperator{\postFix}{Fix^{\text{\tiny $\geq$}}}

\DeclareMathOperator{\Safe}{Safe}

\DeclareMathOperator{\Init}{Init}
\DeclareMathOperator{\id}{id}

\newcommand{\cF}{\mathcal{F}}

\newcommand{\cS}{\mathcal{S}}
\newcommand{\cG}{\mathcal{G}}

\newcommand{\cP}{\mathcal{P}}

\newcommand{\cA}{\mathcal{A}}

\newcommand{\cT}{\mathcal{T}}
\newcommand{\cC}{\mathcal{C}}

\newcommand{\cX}{\mathcal{X}}
\newcommand{\cR}{\mathcal{R}}

\newcommand\Tau{\textsf{T}}

\newcommand{\bV}{\mathbb{V}}
\newcommand{\bQ}{\mathbb{Q}}
\newcommand{\bN}{\mathbb{N}}
\newcommand{\bZ}{\mathbb{Z}}
\newcommand{\bR}{\mathbb{R}}

\newcommand{\dotted}[1]{\mathrel{\dot{#1}}}

\begin{document}

\title{Decidability and Synthesis of Abstract Inductive Invariants}
\author{Francesco Ranzato}

\institute{Dipartimento di Matematica, University of Padova, Italy}
\pagestyle{plain}

\maketitle 

\begin{abstract}
Decidability and synthesis of inductive invariants ranging in a given 
domain play an important role
in many software and hardware verification systems.  
We consider here inductive invariants belonging to an abstract domain $A$
as defined in abstract interpretation, namely, ensuring 
the existence of
the best approximation in $A$ of any system property. In this setting, we study
the decidability of the existence of abstract inductive invariants 
in $A$ of transition systems and their corresponding algorithmic synthesis.  
Our model relies on some general results which relate the existence of abstract inductive
invariants with least fixed points of best correct approximations in $A$ of
the transfer functions of transition systems  and
their completeness properties.   
This approach allows us to derive 
decidability and synthesis results for abstract inductive invariants which are applied to the well-known 
Kildall's constant propagation and Karr's affine equalities abstract domains.
Moreover, we show that a recent general algorithm  
for synthesizing inductive invariants in domains of logical formulae 
can be systematically
derived from our results and generalized to a 
range of algorithms for computing
abstract inductive invariants. 
\end{abstract}

\section{Introduction}
Proof and inference methods based on inductive invariants are widespread in automatic (or semi-automatic)
program and system verification (see, \textit{e.g.},
\cite{chatterjee10,cousot19,CC82,dillig13,kincaid,sagiv20,padon19,sagiv-cav19,ivrii14,dafny1,ma19,mine16,seidl18,padon16}).
The inductive invariant proof method roots at the works of Floyd~\cite{floyd67},
Park~\cite{park1969,park79},  Naur~\cite{naur}, 
Manna et al.~\cite{manna} back to Turing~\cite{turing49}. 
Given a transition system $\cT=\tuple{\Sigma,\tau,\Sigma_0}$, \textit{e.g.} representing a program, where $\tau$ is a transition relation on 
states ranging in $\Sigma$ and
$\Sigma_0\subseteq \Sigma$  is a set of initial states,  
together with a safety property $P\subseteq \Sigma$ to check,  
let us recall that 
a state property $I\subseteq \Sigma$ is an \emph{inductive invariant} for 
$\tuple{\cT,P}$ when:
$\Sigma_0 \subseteq I$, \textit{i.e.}\ the initial states satisfy $I$; 
$I\subseteq P$, \textit{i.e.}\ $I$ entails the safety property $P$; $\tau(I)\subseteq I$, 
\textit{i.e.}\ $I$ is inductive. The \emph{inductive invariant principle}
states that $P$ 
holds for all the reachable states of  $\cT$ 
iff there exists an inductive invariant $I$ for $\tuple{\cT,P}$. In this
explicit form this principle has been probably 
first formulated by Cousot and Cousot \cite[Section~5]{CC82} in 1982
and called ``induction principle for invariance proofs''.
In most cases, verification and inference methods rely on 
inductive invariants $I$ that range in some restricted domain 
$\cA\subseteq
\wp(\Sigma)$, such as a domain of logical formulae (\textit{e.g.}, some separation
logic or some fragment of
first-order logic \cite{padon16}) or a domain of abstract interpretation \cite{CC77,CC79}
(\textit{e.g.}, numerical abstract domains of affine relations or convex polyhedra).
In this scenario, if an inductive invariant $I$ belongs to 
$\cA$ then $I$ is here called an \emph{abstract inductive invariant}.  

\paragraph{\textbf{Main Contributions.}} Our primary goal was to investigate 
whether and 
how the inductive invariant principle can be adapted
 when  inductive invariants are restricted to range in an abstract domain $\cA$.  We formulate this problem in 
general order-theoretical terms and we make the following working assumption: 
$\cA\subseteq
\wp(\Sigma)$ is an abstract domain as defined in abstract interpretation \cite{CC77,CC79}. This means that each state
property $X\in \wp(\Sigma)$ has a \emph{best} over-approximation (w.r.t.\ $\subseteq$) $\alpha_\cA(X)$ 
in $\cA$ and each state transition relation  $\tau$ has
a \emph{best} correct approximation $\tau^\cA$ on the abstract domain $\cA$. Under these hypotheses, we prove an \emph{abstract inductive invariant principle} stating that there exists an abstract inductive invariant in $\cA$ proving a property $P$ of a transition system $\cT$ 
iff the \emph{best abstraction} $\cT^\cA$ 
in $\cA$ of the system $\cT$ allows us to prove $P$. 
The decidability/undecidability question of the existence of abstract inductive invariants 
in some abstract domain $\cA$ for some class of 
transition systems has been recently investigated in a few
significant cases \cite{worrell19,worrell18,mon19,shoham18}.
We show how the abstract inductive invariant principle 
allows us to derive a general 
decidability result on the existence of inductive invariants 
in some abstract domain $\cA$ and to design a 
general algorithm for synthesizing 
the least (w.r.t.\ the order of $\cA$) 
abstract inductive invariant in $\cA$, when this exists, by a least fixpoint computation
in $\cA$.

We also show a related result 
which  is
of independent interest in abstract interpretation: 
the (concrete) inductive invariant principle for a system $\cT$
is equivalent to the abstract inductive invariant principle for $\cT$ on an abstract domain $\cA$ iff \emph{fixpoint completeness} of $\cT$ on $\cA$ holds, \textit{i.e.}, 
the best abstraction in $\cA$ of the reachable states of $\cT$ coincides
with the reachable states of the best abstraction $\cT^\cA$ of $\cT$ on $\cA$.

The decidability/synthesis
of abstract inductive invariants 
in a  domain $\cA$ for some class $\cC$ of systems  
essentially boils down to prove that
the best correct approximation $\tau^{\cA}$ in $\cA$ 
of the transition relation $\tau$ of systems in the class $\cC$ 
is computable. As case studies, we provide two such 
results for Kildall's 
constant propagation  \cite{kildall73}
and Karr's affine equalities \cite{karr76} domains, 
which are well-known and widely used abstract domains 
in numerical program analysis \cite{mine17}. 
As a second application, we design an inductive invariant synthesis algorithm which, 
by generalizing a recent algorithm by Padon et al.~\cite{padon16}
tailored for logical invariants, 
outputs the most abstract (\textit{i.e.}, weakest) inductive invariant in a
domain $\cA$ which satisfies some suitable hypotheses.  
In particular, we show that this synthesis algorithm is obtained 
by instantiating a concrete co-inductive (\textit{i.e.}, based on
a greatest fixpoint computation) fixpoint checking algorithm by Cousot~\cite{cou00} to a domain $\cA$ of abstract  invariants which is \emph{disjunctive}, \textit{i.e.}, 
the abstract least upper bound of $\cA$ does
not lose precision. This generalization allows us to design further related co-inductive algorithms for synthesizing abstract inductive invariants.

\section{Background}\label{sec:background}

\subsection{Order Theory}%
If $X$ is a subset of some universe set $U$ then $\neg X$ denotes the complement of $X$ with respect to $U$ when $U$ is implicitly
given by the context. 
If $f:X\ra Y$ is a function between sets and $S\in \wp(X)$ then $f(S)
\ud \{f(x) \in Y \mid x\in S\}$ denotes the image of $f$ on $S$. A composition of two functions $f$ and $g$ is denoted both by $fg$ or $f\circ g$.
If $\vec{x}\in X^n$ is a vector in a product domain, $j\in [1,n]$ 
and $y\in X$ then $\vec{x}[x_j/y]$ denotes
the vector obtained from $\vec{x}$ by replacing  
its $j$-th component $x_j$ with $y$.  
To keep the notation simple and compact, we use the same symbol for a function/relation 
and its componentwise (\textit{i.e.}\ pointwise) extension on product domains, \textit{e.g.}, if
$\vec{S},\vec{T}\in \wp(X)^n$ then $\vec{S}\subseteq \vec{T}$ 
denotes that for all $i\in [1,n]$, 
$\vec{S}_i\subseteq \vec{T}_i$. Sometimes, to emphasize a pointwise definition, 
a dotted notation can be used such as in $f\dotted{\leq} g$ for
the pointwise ordering between functions.

A quasiordered set \(\tuple{D,\leq}\) (qoset), compactly denoted by $D_\leq$, 
is a set $D$ endowed with a quasiorder (qo) relation $\leq$ on $D$, \textit{i.e.}\ a reflexive and
transitive binary relation. 
A qoset $D_\leq$ satisfies the ascending (resp.\ descending) chain condition (ACC, resp.\ DCC) if  $D$ contains no countably infinite  sequence of distinct elements \(\{x_i\}_{i \in \mathbb{N}}\) such that, for all $i\in\bN$, \(x_i \leq x_{i{+}1}\) (resp. \(x_{i{+}1} \leq x_{i}\)).
An antichain in a qoset $D$ is a subset $X\subseteq D$ such that 
any two distinct elements in $X$ are incomparable for the qo $\leq$. 
A qoset $D_\leq$ is a partially ordered set (poset) when \(\leq\) is antisymmetric.
A subset $X\subseteq D$ of a poset is directed if $X$ is nonempty and every pair of elements in $X$ has an upper bound in $X$.
A poset is a directed-complete partial order (CPO) if it has 
the least upper bound (lub) of all its directed subsets. 
A
complete lattice is a poset having  
the lub of all its arbitrary (possibly empty) subsets (and therefore also having
arbitrary glb's). 
In a complete lattice (or CPO), $\vee$ (or $\sqcup$) and $\wedge$ (or $\sqcap$) denote, resp., lub and glb, 
and $\bot$ and $\top$ denote, resp., least and greatest element.

Let \(P_\leq\) be a poset and  
\(f:P \ra P\). Then, $\Fix(f)\ud \{x\in P \mid f(x)=x\}$, 
$\preFix(f)\ud \{x\in P \mid f(x)\leq x\}$, 
$\postFix(f)\ud \{x\in P \mid f(x)\geq x\}$, and $\lfp(f)$,
$\gfp(f)$ denote, resp., the least and greatest fixpoint 
in $\Fix(f)$, when they exist. 
Let us recall Knaster-Tarski fixpoint theorem: if $\tuple{C,\leq,\vee,\wedge}$ 
is a complete lattice and $f:C\ra C$ is monotonic (\textit{i.e.}, $x\leq y$ implies $f(x)\leq f(y)$) then 
$\tuple{\Fix(f),\leq}$ is a complete lattice, $\lfp(f)=\wedge\preFix(f)$
and $\gfp(f)=\vee\postFix(f)$. Also, Knaster-Tarski-Kleene 
fixpoint theorem states that if $\tuple{C,\leq,\vee,\bot}$ is a CPO 
and $f:C\ra C$ is Scott-continuous (\textit{i.e.}, $f$ preserves lub of directed subsets) 
then $\lfp(f)=\vee_{i\in \bN} f^i(\bot)$, where, for all $x\in C$ and $i\in \bN$, 
$f^0(x)\ud x$ and $f^{i+1}(x)\ud f(f^i(x))$;
dually, if $\tuple{C,\leq,\wedge,\top}$ is a dual-CPO and 
$f:C\ra C$ is Scott-co-continuous then $\gfp(f)=\wedge_{i\in \bN} f^i(\top)$. 
A function $f:C\ra C$ on a complete lattice is 
additive when it preserves arbitrary lubs.

\subsection{Abstract Domains}
\label{appendix:background}
Let us recall some basic notions on closures and Galois connections which are commonly used in abstract interpretation \cite{CC77,CC79} to define abstract domains 
(see, \textit{e.g.}, \cite{mine17,rival}).  
Closure operators and Galois connections are equivalent notions
and are both used for 
defining the notion of approximation in abstract interpretation, where closure operators bring the advantage of defining abstract domains independently of a specific representation which is required by 
Galois connections. 

An upper closure operator (uco), or simply upper closure, on a poset $C_\leq$ 
is a function \(\mu:C\ra C\) which is:
monotonic, idempotent and
extensive, \textit{i.e.}, \(x \leq \mu(x)\) for all \(x \in C\).
Dually, a lower closure operator (lco) 
$\eta:C\ra C$ is monotonic, idempotent and reductive, \textit{i.e.},
\(\eta(x) \leq x\) for all \(x \in C\).
The set of all upper/lower closures on $C_\leq$ 
is denoted by $\uco(C_\leq)$/$\lco(C_\leq)$.
We write \(c \in \mu(C)\), or simply \(c \in \mu\), to denote that  
there exists \(c' \in C\) such that \(c = \mu(c')\), and 
we recall that this happens iff $\mu(c) = c$. 
Let us also recall that $\tuple{\mu(C),\leq}$ is closed under glb
of arbitrary subsets and, conversely, 
$X\subseteq C$ is the image of some $\mu\in \uco(C)$ iff
$X$ is closed under glb
of all its subsets, and in this case 
$\mu(c)=\wedge \{c'\in X \mid c \leq c'\}$ holds. 
 Dually, $X\subseteq C$ is closed under arbitrary lub of its subsets iff $X$ is the image a lower closure
$\eta\in \lco(C)$, and in this case $\eta(c)=\vee \{c'\in X \mid c'\leq c\}$.  
In abstract interpretation, a closure \(\mu\in \uco(C_\leq)\) on a concrete domain $C_\leq$ plays
the role of an abstract domain having best approximations: $c\in C$ is approximated by any $\mu(c')$ such that $c\leq \mu(c')$ and $\mu(c)$ is the best approximation of $c$ in $\mu$ because $\mu(c)=\wedge\{\mu(c') \mid c'\in C,\, c\leq \mu(c')\}$.

A Galois Connection (GC) (also called adjunction) between two posets \(\tuple{C,\leq_C}\), called concrete domain, and \(\tuple{A,\leq_A}\), called abstract domain, consists of two maps \(\alpha: C\ra A\) and \(\gamma: A\ra C\) such that \(\alpha(c)\leq_A a \:\Lra\: c\leq_C \gamma(a)\) holds. A GC is called Galois insertion (GI) when $\alpha$ is surjective or,
equivalently, $\gamma$ is injective. Any GC can be transformed into a GI simply by removing useless elements in $A\smallsetminus \alpha(C)$ from the abstract domain $A$. 
A GC/GI is denoted by $(C_{\leq_C},\alpha,\gamma, 
A_{\leq_A})$.
The function $\alpha$ is called the left-adjoint of $\gamma$, and, dually, 
$\gamma$ is called the right-adjoint of $\alpha$. This terminology is justified by the fact that if
$\alpha:C\ra A$ 
admits a right-adjoint $\gamma:A\ra C$ then this is unique, and this dually holds for left-adjoints.
GCs and ucos are equivalent notions because
any GC $\cG=(C,\alpha,\gamma, A)$ induces a closure $\mu_\cG\ud \gamma\circ \alpha \in \uco(C)$, any $\mu\in \uco(C)$ induces a GI $\cG_\mu\ud
(C,\mu,\lambda x.x, \mu(C))$, and these two transforms are inverse of 
each other.  

\subsection{Transition Systems}
Let $\cT=\tuple{\Sigma,\tau}$ be a transition system where 
$\Sigma$ is a set of states
and  $\tau \subseteq \Sigma\times \Sigma$ is a transition relation. 
As usual, a transition relation can be equivalently defined by one
of the following transformers of type $\wp(\Sigma)\ra
\wp(\Sigma)$:
\begin{align*}
&\!\pre(X)\ud \{s\in \Sigma\mid \exists s'\in X. (s,s')\in \tau\}
\;\pret(X)\ud \{s\in \Sigma\mid \forall s'.(s,s')\in \tau \Ra\! s'\in X\}\\
&\!\post(X)\!\ud \{s'\!\in \Sigma\mid \exists s\in X. (s,s')\in \tau\}
\;\postt(X)\ud\{s'\!\in \Sigma\mid \forall s.(s,s')\in \tau \Ra\! s\in X\}
\end{align*}
We will equivalently specify a transition system by one of the above 
transformers (typically $\post$) in place of the transition relation $\tau$. 
Let us also recall (see \textit{e.g.}~\cite{CC99}) that $\tuple{\pre,\postt}$ and $\tuple{\post,\pret}$ are pairs of adjoint functions.
The 
set of reachable states of $\cT$ from a set of initial states $\Sigma_0\subseteq \Sigma$ is given by $\Reach[\cT,\Sigma_0]\ud 
\lfp(\lambda X. \Sigma_0 \cup \post(X))$, and  
$\cT$ satisfies a safety property $P\subseteq
\Sigma$ when $\Reach[\cT,\Sigma_0]\subseteq P$ holds.

\subsection{Inductive Invariant Principle}
Given a transition system $\cT=\tuple{\Sigma,\tau}$, a set of states $I\in \wp(\Sigma)$ is an \emph{inductive invariant} for $\cT$ w.r.t.\ $\tuple{\Sigma_0,P}\in \wp(\Sigma)^2$ when: 
(i) $\Sigma_0\subseteq I$; (ii) $\post(I)\subseteq I$; (iii) $I\subseteq P$. 
An inductive invariant $I$ allows us to prove that $\cT$ is safe, \textit{i.e.}\ 
$\Reach[\cT,\Sigma_0]\subseteq P$,  
by the \emph{inductive invariant principle} (a.k.a.\ fixpoint induction principle), a consequence of Knaster-Tarski fixpoint theorem: If $C_\leq$ is a complete lattice, $c'\in C$ and 
$f:C\ra C$ is monotonic then 
\begin{equation}\label{pii}
\lfp(f)\leq c' \Lra \exists i\in C. f(i)\leq i \wedge i\leq c'
\end{equation} 
In particular, given $c,c'\in C$, since $c\vee_C f(i)\leq i$ iff 
$c\leq i \wedge f(i)\leq i$, it turns out that:
\begin{equation}\label{inv}
\lfp(\lambda x. c\vee_C f(x))\leq c' \Lra \exists i\in C. c\leq i \wedge 
f(i)\leq i \wedge i\leq c'
\end{equation} 
One such $i\in C$ such that $c\leq i\, \wedge\, 
f(i)\leq i\, \wedge \, i\leq c$ is called an inductive invariant of $f$
for $\tuple{c,c'}$. 
Hence, \eqref{inv} is applied to the function $\lambda 
X. \Sigma_0 \cup \post(X): \wp(\Sigma)\ra \wp(\Sigma)$, which 
is monotonic on the complete lattice $\wp(\Sigma)_\subseteq$, so that
$\lfp(\lambda 
X. \Sigma_0 \cup \post(X))\subseteq P$ holds iff  there exists  
an inductive invariant $I$ for $\cT$ w.r.t.\ $\tuple{\Sigma_0,P}$.
In most interesting contexts for defining 
transition systems, 
the decision problem of the existence of a (concrete) inductive invariant for
a class of transition systems w.r.t.\ 
a set of initial states and some safety property 
is undecidable. %

\section{Abstract Inductive Invariants}

Padon et al.\ \cite{padon16}, Hrushovski et al.~\cite{worrell18}, 
Fijalkow et al.~\cite{worrell19}, Monniaux~\cite{mon19}, 
Shoham \cite{shoham18}, among the others, 
consider a notion of abstract inductive invariant and study 
the corresponding decidability/undecidability and synthesis problems. 
The common approach of this stream of works 
consists in restricting  the range of inductive invariants
from a concrete domain $C$ to some abstraction $A_C$ of $C$. 
In a basic and general form, a domain $A_C$ of abstract invariants 
is simply a subset 
of $C$. Let us formalize abstract inductive invariants in general 
order-theoretic terms. Given a class $\cC$ of complete
lattices  and, for all $C\in \cC$, 
a class of functions $\cF_C \subseteq C\ra C$, a 
set of initial properties $\Init_C\subseteq C$,
a set of safety properties $\Safe_C\subseteq C$, and
some abstract domain $A_C \subseteq C$, a first problem is the  
decidability of the following decision question:
\begin{equation}\label{ainv}
\forall C\in \cC.\forall f\in \cF_C. \forall c\in \Init_C. 
\forall c'\in \Safe_C. 
\exists i\in^?\!\! A_C.\: c\leq i \wedge 
f(i)\leq i \wedge i\leq c'
\end{equation} 
where one such $i\in A_C$ is called an \emph{abstract inductive invariant} 
for $f$ and $\tuple{c,c'}\in C^2$. 
On the other hand, the synthesis problem consists in designing algorithms 
which output abstract inductive invariants in $A_C$ or notify
that no inductive invariant in $A_C$ exists.  

Given a transition system $\cT=\tuple{\Sigma,\tau}$ whose 
successor transformer is $\post$,  
the problem \eqref{ainv} is instantiated to 
$C_\leq=\wp(\Sigma)_\subseteq$, $f=\post(X)$, 
$c=\Sigma_0\in \wp(\Sigma)$ set of initial states and $c'=P\in \wp(\Sigma)$ safety property.  
When the transition system is generated by 
some imperative program, $\Sigma_0$ are the states of some initial control
node and $P$ is a safety property given by the states which are not in some bad control node, 
abstract inductive invariants are called \emph{separating invariants} and  
the decision problem~\eqref{ainv} is called \emph{Monniaux problem} 
by Fijalkow et al.~\cite{worrell19}, because this was first 
formulated by Monniaux \cite{mon17,mon19}.

\subsection{Abstract Inductive Invariant Principle}
Our working assumption is that in problem~\eqref{ainv} 
invariants $i$ range in an abstract domain $A$
as defined in abstract interpretation \cite{CC77,CC79}.  
\begin{assumption}\label{assu}
$\tuple{A,\leq_A}$ is an abstract domain 
of the complete lattice 
$\tuple{C,\leq_C}$ which has best approximations, \textit{i.e.}, one of these 
two equivalent assumptions is satisfied: 
\begin{enumerate}[{\rm (i)}]
\item $(C_{\leq_C},\alpha,\gamma,A_{\leq_A})$ is a Galois insertion;
\item $\tuple{A,\leq_A}=\tuple{\mu(C),\leq_C}$ for some upper closure $\mu\in \uco(C_{\leq_C})$. \qed
\end{enumerate}
\end{assumption} 

Under Assumption~\ref{assu}, 
let us recall that if $f:C\ra C$ is a concrete monotonic function then 
the mappings $\alpha f \gamma: A\ra A$, for the case of GIs, and $\mu f: \mu(C) \ra \mu(C)$, for the case of ucos, are called
\emph{best correct approximation} (bca) in $A$ of $f$. This is justified by the observation that 
an abstract function $f^\sharp:A\ra A$ (or $f^\sharp:\mu(C)\ra \mu(C)$ for ucos) 
is a correct (or sound) 
approximation of $f$ 
when $\alpha f\gamma \dotted{\leq}_A f^\sharp$ (or $\mu f \dotted{\leq}_C f^\sharp$ for ucos) holds. Our first result is an \emph{abstract} 
inductive invariant principle which restricts the invariants of $f$ in \eqref{pii} to 
those ranging in an abstract domain $A$: when the abstract domain $A$ is 
specified by a GI, 
this means that $a\in A$ is an abstract invariant of $f$ when $f\gamma(a)\leq_C \gamma(a)$ holds; when the abstract domain is a closure $\mu\in \uco(C)$, this means 
that $a\in\mu\subseteq C$ 
is an abstract invariant of $f$ when $fa\leq_C a$ holds.

\begin{lemma}[\textbf{Abstract Inductive Invariant Principle}]\label{char}\ \
Let $(C_{\leq_C},\alpha,\gamma,A_{\leq_A})$ be a GI. 
For all $c'\in C$ and $a'\in A$: 
\begin{enumerate}[{\rm (a)}]
\item 
$\gamma(\lfp(\alpha f \gamma)) \leq_C c' \Lra \exists a\in A.\: f\gamma(a) \leq_C \gamma(a) \wedge \gamma(a)\leq_C c'$;
\item  
$\lfp(\alpha f \gamma) \leq_A a' \Lra \exists a\in A.\: f\gamma(a) \leq_C \gamma(a) \wedge \gamma(a)\leq_C \gamma(a')$.
\end{enumerate}
\end{lemma}
\begin{proof}
Let us first recall that in  a GI, for all $a,a'\in A$, $a\leq_A a' \Lra \gamma(a)\leq_C \gamma(a')$ holds.

\noindent
(a) $(\La)$ We have that:
\begin{align*}
\exists a\in A.\: f\gamma(a) \leq_C \gamma(a) \wedge \gamma(a)\leq_C c' & \Lra \quad[\text{by GC}]\\
\exists a\in A.\: \alpha f\gamma(a) \leq_A a \wedge \gamma(a)\leq_C c'  & \Ra \quad[\text{by $Fx \leq x \Ra \lfp(F) \leq x$}]\\
\exists a\in A.\: \lfp(\alpha f\gamma) \leq_A a \wedge \gamma(a)\leq_C c'  & \Lra\quad[\text{by GI}]\\
\exists a\in A.\: \gamma(\lfp(\alpha f\gamma)) \leq_C \gamma(a) \wedge \gamma(a)\leq_C c'  &\Ra \quad[\text{by transitivity]}\\
 \gamma(\lfp(\alpha f \gamma)) \leq_C c' &
\end{align*}
$(\Ra)$ 
Define $a\ud \lfp(\alpha f\gamma)\in A$. It turns out that $\alpha f\gamma (a) \leq_A a$ so that, by GC, 
$f\gamma(a) \leq_C \gamma(a)$, and, by hypothesis, 
$\gamma(a)\leq_C c'$. 

\noindent
(b) It turns out that:
\begin{align*}
\exists a\in A.\: f\gamma(a) \leq_C \gamma(a) \wedge \gamma(a)\leq_C \gamma(a') & \Lra 
\quad \text{[By Lemma~\ref{char}~(a)]}\\
\gamma(\lfp(\alpha f \gamma)) \leq_C \gamma(a') & \Lra \quad [\text{by GI}]\\
\lfp(\alpha f \gamma) \leq_A a' & \tag*{\qed}
\end{align*} 
\end{proof}

\noindent
It is worth stating Lemma~\ref{char}~(a) in an equivalent form for an abstract domain represented by a closure $\mu\in \uco(C)$: $\lfp(\mu f) \leq_C c' \Lra \exists a\in \mu.\: fa \leq_C a \wedge a\leq_C c'$.

Let us observe that point~(b) is an easy consequence of  
point~(a), because, by surjectivity of $\alpha$ in GIs, for all $a'\in A$, there exists some $c'\in C$ such that $a'=\alpha(c')$, and $\gamma(\lfp(\alpha f \gamma)) \leq_C \gamma(\alpha(c')) \Lra
 \lfp(\alpha f \gamma) \leq_A \alpha(c')$ holds. Moreover, point~(b)
easily follows from the inductive invariant
principle~\eqref{pii} for the bca $\alpha f \gamma: A\ra A$.  
On the other hand, point~(a) cannot be obtained from (b), \textit{i.e.}\ (a) is \emph{strictly stronger} than~(b), because (a) allows us to  prove concrete 
properties $c'\in C$ which are not exactly represented by $A$ (\textit{i.e.}, $c'\not\in \gamma(A)$) by 
abstract inductive invariants in $A$. 
This is shown
by the following tiny example. 

\begin{examplebf}
Consider a 4-point chain $C = \{1<2<3<4\}$,
the function 
$f:C\ra C$ defined by $\{1 \mapsto 1; 
2\mapsto 2; 3\mapsto 4; 4\mapsto 4\}$,
and the abstraction $A=\{2,4\}$ with $\gamma=\id$ and $\alpha=
\{1 \mapsto 2;
2\mapsto 2; 3\mapsto 4; 4\mapsto 4\}$. Here, we have that  
$\alpha f \gamma =\{2 \mapsto 2;
4\mapsto 4\}$ and $\lfp(\alpha f \gamma)=2$. 
In this case, Lemma~\ref{char}~(b) allows us  
to prove all the abstract properties 
$a'\in A$ by abstract inductive invariants, 
while Lemma~\ref{char}~(a) allows us to prove an additional
concrete property $3\in C\smallsetminus\gamma(A)$, which is not exactly represented by $A$,
by an abstract inductive invariant,  and this would not be possible
by resorting to Lemma~\ref{char}~(b).
Also, $\gamma(\lfp(\alpha f \gamma))\not\leq 1$ holds, thus, by Lemma~\ref{char}~(a), the concrete property $1$ cannot be proved by an abstract inductive invariant in $A$, 
whereas Lemma~\ref{char}~(b) does not allow us to infer this.  
\qed
\end{examplebf}

Lemma~\ref{char}~(b) tells us that 
the existence of an abstract inductive invariant of $f$ proving 
an abstract property $a'\in A$ is equivalent to 
the fact that the least fixpoint of the best correct approximation $\alpha f \gamma$
entails
$a'$. This formalizes for an abstract domain satisfying 
Assumption~\ref{assu}
an observation in 
\cite[Section~1]{worrell19} stating
that ``\emph{the existence of some abstract inductive 
invariant for $\alpha f \gamma$ proving $a'$ 
is equivalent to whether the strongest abstract 
invariant $\lfp(\alpha f \gamma)$ entails $a'$}'', \textit{i.e.}\ is inductive. 
If, instead, we aim at proving \emph{any} concrete property $c'\in C$,
possibly not in $\gamma(A)$, 
by an abstract inductive invariant
then
Lemma~\ref{char}~(a) states that this is equivalent to the stronger condition $\gamma(\lfp(\alpha f \gamma)) \leq_C c'$.

As a consequence of Lemma~\ref{char}~(a) we derive the following characterization of the problem~\eqref{ainv}.

\begin{corollary}\label{coro-three}
Let $\cF\subseteq C\ra C$ and $\Init,\Safe\subseteq C$. 
The Monniaux decision 
problem $\forall f\in \cF.\forall c\in \Init.\forall c'\in \Safe. 
\exists a\in^?\!\! A.\: c \leq_C \gamma(a) \wedge 
f\gamma(a) \leq_C \gamma(a) \wedge \gamma(a)\leq_C c'$ is decidable iff
the decision problem $\forall f\in \cF.\forall c\in \Init.\forall c'\in \Safe. \gamma(\lfp(\lambda x\in A. \alpha(c) \vee_A  \alpha f \gamma(x))) \leq_C^? c'$ is decidable.
\end{corollary}
\begin{proof} 
By Lemma~\ref{char}~(a), because 
$\lambda x\in A. \alpha(c) \vee_A  \alpha f \gamma(x) = 
\lambda x\in A. \alpha(c \vee_C f \gamma(x))$ is the best correct approximation
of $\lambda x\in C. c \vee_C f(x)$. 
\qed
\end{proof}

Moreover, as a consequence of Lemma~\ref{char}~(b) we obtain the following abstract invariant synthesis algorithm. 

\begin{corollary}\label{coro-algo}
Assume that the lub $\vee_A: A\times A \ra A$ and the bca 
$\alpha f \gamma : A \ra A$ are computable, the partial order $\leq_A^?$ is decidable and $A$ is an ACC CPO.
For all $c\in C$ such that $\alpha(c)$ is computable and $a'\in A$, the following procedure $\AInv(f,A,c,a')$:
\vspace*{-2.5pt}
{\rm 
\begin{align*}
&i:=\alpha(c);\\[-2.5pt]
&\textbf{while~}i\leq_A a' 
\textbf{~do~} \{\textbf{~if~}\alpha f\gamma(i) \leq_A i \textbf{~return~}i;
\textbf{~else~} i:= \alpha f\gamma(i);\}\\[-2.5pt]
&\textbf{return~}\textit{no abstract inductive invariant for $f$ and $\tuple{c,\gamma(a')}$};
\end{align*}
}
\\[-12.5pt]
\noindent
is a terminating algorithm which outputs the least abstract inductive invariant
for $f$ and $\tuple{c,\gamma(a')}$, when one such abstract 
inductive invariant exists, otherwise 
outputs ``no abstract inductive invariant''. 
\end{corollary}
\begin{proof}
The hypotheses guarantee that the procedure $\AInv$ is a terminating algorithm, 
in particular because the sequence of computed iterates $i$ is an ascending chain in $A$.  
If the 
algorithm $\AInv$ outputs $i$ then $i=\lfp(\lambda a. \alpha(c) \vee_A \alpha f \gamma(a))\leq_A a'$, so that $i = \wedge \{a\in A \mid \alpha(c) \leq_A i,\, 
\alpha f \gamma(i)\leq_A i,\, i\leq_A a'\}$, that is, $i$ 
is the least inductive invariant in $A$ 
for $f$ and $\tuple{c,\gamma(a')}$. 
If the algorithm $\AInv$ outputs 
``no abstract inductive invariant for $f$ and $\tuple{c,\gamma(a')}$'' 
then there exists 
$j\in \bN$ such that 
$(\lambda a. \alpha(c) \vee_A \alpha f \gamma(a))^j(\bot_A) \not\leq_A a'$, 
so that $\lfp(\lambda a. \alpha(c) \vee_A \alpha f \gamma(a))\not\leq_A a'$, 
that is, there exists no inductive invariant in $A$ 
for $f$ and $\tuple{c,\gamma(a')}$. 
\qed\end{proof}

\subsection{Fixpoint Completeness in Abstract Interpretation}\label{compl-sec}
Soundness in abstract interpretation (more in general, in static analysis) 
is a mandatory requirement
stating that no false negative can occur, that is, 
abstract fixpoint computations correctly (over-)approximate
the corresponding concrete semantics:  
if $f:C\ra C$ and $f^\sharp:A\ra A$ are the concrete and 
abstract monotonic transformers then \emph{fixpoint soundness} means that 
$\alpha(\lfp(f))\leq_A \lfp(f^\sharp)$ holds, so that 
a positive abstract check $\lfp(f^\sharp) \leq_A a'$ 
entails that $\gamma(a')$ concretely holds, \textit{i.e.}, $\lfp(f) \leq_C \gamma(a')$. 
Fixpoint soundness is usually obtained as a consequence of 
\emph{pointwise soundness}: if $f^\sharp$
is a pointwise correct 
approximation of $f$, \textit{i.e.}\ $\alpha f \dotted{\leq}_A f^\sharp \alpha$, then $\alpha(\lfp(f))\leq_A 
\lfp(f^\sharp)$ holds. 

While soundness is indispensable, 
completeness in abstract interpretation encodes an ideal
situation where no false positives (also called false alarms) 
arise: \emph{fixpoint completeness} means that 
$\alpha(\lfp(f))= \lfp(f^\sharp)$ holds, so that
$\lfp(f^\sharp) \not\leq_A a'$ entails
$\lfp(f) \not\leq_C \gamma(a')$. One can also consider a \emph{strong
fixpoint completeness} requiring that $\lfp(f)= \gamma(\lfp(f^\sharp))$, 
so that $\lfp(f^\sharp) \not\leq_A \alpha(c')$ entails
$\lfp(f) \not\leq_C c'$. However, it should be remarked that 
$\lfp(f)= \gamma(\lfp(f^\sharp))$ is much stronger 
than $\alpha(\lfp(f))= \lfp(f^\sharp)$ since it 
means that the concrete lfp  is 
precisely represented by the 
abstract lfp.

It is important to remark that 
if $f^\sharp$ is a pointwise correct 
approximation of $f$ and fixpoint completeness for 
$f^\sharp$ holds then since $\alpha(\lfp(f)) \leq_A \lfp(\alpha f \gamma)
\leq_A \lfp(f^\sharp)$ always holds, one also obtains that  
$\alpha(\lfp(f)) = \lfp(\alpha f \gamma)=\lfp(f^\sharp)$ holds, namely, the best
correct approximation  $\alpha f \gamma$ is fixpoint complete as well. 
This means that the possibility (or impossibility) of defining an approximate transformer 
$f^\sharp:A\ra A$ on $A$ which is 
fixpoint complete does not depend on the specific definition of 
$f^\sharp$ but is instead an intrisic  
\emph{property of the abstract domain} $A$ w.r.t.\ the concrete transformer 
$f$, as formalized by the
equation $\alpha(\lfp(f)) = \lfp(\alpha f \gamma)$. 
Moreover, fixpoint completeness 
is typically proved as a by-product of \emph{pointwise completeness} 
$\alpha f = f^\sharp \alpha$, and if $f^\sharp$ is 
pointwise complete
then it turns out that 
$f^\sharp = \alpha f \gamma$, that is, $f^\sharp$ actually is
the bca of $f$. This justifies why, without loss of generality,
we can consider 
fixpoint and pointwise completeness 
of bca's $\alpha f \gamma$ only, \textit{i.e.}, a property of 
abstract domains.

\subsection{Characterizing Fixpoint Completeness by Abstract Inductive Invariants}\label{sec-char}
We show that the abstract inductive invariant principle
is closely related to fixpoint completeness. More precisely, we provide
an answer to the following question: in the abstract inductive invariant principle
as stated by Lemma~\ref{char},  
can we replace $\lfp(\alpha f \gamma)$ with $\alpha(\lfp(f))$?
This question is settled by the following result. 

\begin{theorem}%
\label{theo-char}
Let $(C_{\leq_C},\alpha,\gamma,A_{\leq_A})$ be a GI.
\begin{enumerate}[{\rm (a)}]
\item
$\lfp(f) = \gamma(\lfp(\alpha f \gamma))$ iff
$\forall c'\in C. \big(\! \lfp(f) \leq_C c' \Lra \exists a\in A.\: f\gamma(a) \leq_C \gamma(a) \wedge \gamma(a)\leq_C c'\big)$;
\item
$\alpha(\lfp(f)) = \lfp(\alpha f \gamma)$ iff
$\forall a'\in A. \big(\!\lfp(f) \leq_C \gamma(a') \Lra \exists a\in A.\: f\gamma(a) \leq_C \gamma(a) \wedge \gamma(a)\leq_C \gamma(a')\!\big)$.
\end{enumerate}
\end{theorem}
\begin{proof} 

\noindent
{\rm (a)}
$(\Ra)$:  By Lemma~\ref{char}~(a). 

\noindent
$(\La)$: Since $\lfp(f) \leq_C \lfp(f)$ holds, we have that 
$\exists a\in A.\: f\gamma(a) \leq_C \gamma(a) \wedge \gamma(a)\leq_C \lfp(f)$. 
Thus, by Lemma~\ref{char}~(a), $\gamma(\lfp(\alpha f \gamma)) \leq_C \lfp(f)$ follows.
On the other hand, $\lfp(f)\leq  \gamma(\lfp(\alpha f \gamma))$ always holds
because the pointwise correctness of 
$\alpha f \gamma$ implies
$\alpha(\lfp(f))\leq_A  \lfp(\alpha f \gamma)$, hence, by GC, $\lfp(f)\leq  \gamma(\lfp(\alpha f \gamma))$ follows.

\noindent
{\rm (b)} $(\Ra)$: By Lemma~\ref{char}~(b) because $\lfp(\alpha f \gamma) \leq_A a' \Lra
\alpha(\lfp(f)) \leq
a'\Lra \lfp(f) \leq_C \gamma(a')$. 

\noindent
$(\La)$: We consider $a'\ud \alpha(\lfp(f))$, so that, 
$\lfp(f) \leq_C \gamma(a')$ holds and
by the equivalence of the hypothesis, 
$\exists a\in A.\: f\gamma(a) \leq_C \gamma(a) \wedge \gamma(a)\leq_C \gamma\alpha(\lfp(f))$ holds. 
This implies, by GI, that $\exists a\in A.\: \alpha f\gamma(a) \leq_C a \wedge a\leq_A \alpha(\lfp(f))$. 
By the inductive invariant principle \eqref{pii}, this implies that (actually, is equivalent to) $\lfp(\alpha f\gamma )\leq \alpha(\lfp(f))$. 
Furthermore, $\alpha(\lfp(f)) \leq_A \lfp(\alpha f\gamma )$ always holds, therefore proving that $\alpha(\lfp(f)) = \lfp(\alpha f \gamma)$.
\qed\end{proof}

\noindent
Theorem~\ref{theo-char}~(b) can be stated by means of ucos as follows:
if $\mu\in \uco(C)$ then $\mu(\lfp(f)) = \lfp(\mu f)$ iff
$\forall a'\in \mu. (\lfp(f) \leq_C a' \Lra \exists a\in \mu .\: fa \leq_C a \wedge a\leq_C a')$.

The above result can be read as follows. Since,
by the inductive invariant principle~\eqref{pii}, $\lfp(f)\leq_C c'$ iff 
there exists a (concrete) inductive invariant proving $c'$, 
it turns out that Theorem~\ref{theo-char}~(a) states that, for all $c'\in C$,  
the existence of an \emph{abstract} inductive invariant proving $c'$ is equivalent 
to the existence of \emph{any} inductive invariant proving $c'$ 
iff fixpoint completeness holds. In other terms, the (concrete) inductive invariant principle
is equivalent to the abstract inductive invariant principle iff fixpoint completeness holds.
This result is of independent interest in abstract interpretation, 
since it provides a new characterization of 
the key property of fixpoint completeness of abstract domains \cite{grs00}.

A further interesting characterization of fixpoint 
completeness is as follows.

\begin{lemma}\label{lemma-4}
$\alpha(\lfp(f)) = \lfp(\alpha f \gamma) \Lra \exists a\in A. f \gamma(a) \leq_C \gamma(a) \wedge \gamma(a) \leq_C \gamma \alpha (\lfp(f))$.
\end{lemma}
\begin{proof} 
\begin{align*}
\exists a\in A. f \gamma(a) \leq_C \gamma(a) \wedge \gamma(a) \leq_C \gamma \alpha (\lfp(f)) &\Lra \;\:\text{[by GI]}\\
\exists a\in A. \alpha f \gamma(a) \leq_A a \wedge a \leq_A \alpha (\lfp(f)) &\Lra \;\: \text{[by \eqref{pii} for $\alpha f \gamma$]}\\
\lfp(\alpha f \gamma) \leq_A \alpha (\lfp(f)) &\Lra 
\;\:\text{[as $\alpha(\lfp(f)) \leq_A \lfp(\alpha f\gamma )$]}\\
\alpha(\lfp(f)) = \lfp(\alpha f \gamma) &\tag*{\qed}
\end{align*}
\end{proof}

\noindent
As a consequence, 
fixpoint 
completeness for $f$ does not hold in $A$ iff the abstract property $\alpha(\lfp(f))\in A$ cannot be
proved by 
an abstract inductive 
invariant in $A$

\begin{examplebf}
Consider the finite chain $C\ud \{1 < 2 < 3\}$ and the monotonic 
concrete 
function $f:C\ra C$ defined by $f\ud \{1\mapsto 1;\, 2\mapsto 3;\, 3\mapsto 3\}$.

\noindent
Consider
the uco $\mu\ud \{2,3\}$, \textit{i.e.}, $\mu=\{1\mapsto 2;\, 2\mapsto 2;\, 3 \mapsto 3\}$.
Hence, $ \mu f =\{1\mapsto 2;\, 2\mapsto 3;\, 3\mapsto 3\}$ so that
fixpoint completeness does not hold because $\mu(\lfp(f))=\mu(1)=2 < 3=\lfp(\mu f)$. 
Thus, in accordance with Lemma~\ref{lemma-4}, 
it turns out that $\mu(\lfp(f))=2$ cannot be inductively proved in the abstraction $\mu$.
In fact, $f(2)\not \leq 2$, while $f(3)\leq 3$ but $3\not \leq \mu(\lfp(f))$.
 
\noindent
Consider
instead the uco $\mu\ud \{1,3\}$, \textit{i.e.}, $\mu=\{1\mapsto 1;\, 2\mapsto 3;\, 3 \mapsto 3\}$, so that $ \mu f =\{1\mapsto 1;\, 2\mapsto 3;\, 3\mapsto 3\}$. 
Here, $\mu(\lfp(f))=\mu(1)=1=\lfp(\mu f)$, therefore
fixpoint completeness holds.   
Thus, by the uco version of Theorem~\ref{theo-char}~(b), any valid abstract invariant of $f$ can be inductively proved: in fact, $1,3\in \mu$ are valid abstract invariants of $f$ and are both inductive.  
\qed
\end{examplebf}

\subsection{When Safety = Abstract Invariance?}\label{app-two}
Padon et al.~\cite[Section~9]{padon16} in their investigation on the decidability of inferring inductive invariants
state that ``\emph{Usually completeness for abstract interpretation means that the abstract domain is precise enough to prove all interesting safety properties, \textit{e.g.}, {\rm \cite{grs00}}. In our terms, this means that $\safe=\inv$, that is, that all safe programs have an inductive invariant expressible in the abstract domain.}'' 
As a by-product of the results in Section~\ref{sec-char}, 
we are able to give a 
formal justification and statement of this informal characterization of completeness.

Let $\cF \subseteq C\ra C$ be a class of monotonic functions, 
$\cS\subseteq C$ be some set of safety
properties and $\cA \subseteq C$ be an abstract domain of program properties. Let us define:
\begin{align*}
&\safe[\cF,\cS] \ud \{\tuple{f,s}\in \cF\times \cS \mid \lfp(f) \leq_C s\}\\
&\inv[\cF,\cS,\cA] \ud \{\tuple{f,s}\in \cF\times \cS \mid \exists 
a\in \cA.\: fa \leq_C a \wedge a\leq_C s\}
\end{align*}
so that in our model 
$\safe[\cF,\cS]$ and $\inv[\cF,\cS,\cA]$ play the role of, resp., ``safe programs'' and ``programs having 
an inductive invariant expressible in $\cA$''. 
As a consequence of Theorem~\ref{theo-char}, 
we derive the following characterization.

\begin{corollary}\label{coro19}
Assume that $\cA$ satisfies Assumption~\ref{assu} for
some GI $\tuple{C,\alpha,\gamma,\cA}$.
\begin{enumerate}[{\rm (a)}]
\item
Assume that $\cS\subseteq \cA$. Then,
 $\safe[\cF,\cS] =\inv[\cF,\cS,\cA]$ iff $\forall f\in \cF. \alpha(\lfp(f))=\lfp(\alpha f \gamma)$.
 \item
$\safe[\cF,\cS] =\inv[\cF,\cS,\cA]$ iff $\forall f\in \cF. \lfp(f)=\gamma(\lfp(\alpha f \gamma))$.
\end{enumerate}
\end{corollary}
\begin{proof} 
(a) follows by Theorem~\ref{theo-char}~(a), since $\cS\subseteq \cA$ is assumed to hold.
(b) follows by Theorem~\ref{theo-char}~(b).
\qed
\end{proof}

Corollary~\ref{coro19} therefore provides a precise equivalence
of the $\safe =^? \inv$ problem, as stated by Padon et al.~\cite{padon16}, 
with fixpoint completeness (strong fixpoint completeness, in case (b)) in abstract interpretation.

\section{Abstract Inductive Invariants of Programs}

We consider transition systems as represented by a finite control flow graph (CFG) of an 
imperative program.
A program is a tuple $\cP=\tuple{Q,n,\bV,\Tau,\sra}$ where $Q$ is 
a finite set of control nodes (or program points),
$n\in \bN$ is the number of program variables
of type $\bV$ (\textit{e.g.}, $\bV=\bZ,\bQ,\bR$), 
$\Tau$ is a finite set of (possibly nondeterministic) 
transfer functions of type $\bV^n\ra \wp(\bV^n)$,\:
$\mathrel{\sra}\;\subseteq Q\times\Tau\times Q$ 
is a (possibly nondeterministic) 
control flow relation, 
where $q\rat{t}q'$ denotes a flow transition with transfer function $t\in \Tau$. A program $\cP$ therefore defines a transition system $\cT_\cP=\tuple{\Sigma,\tau}$ where $\Sigma\ud Q\times \bV^n$ is the set of states and 
the transition relation $\tau\subseteq \Sigma \times \Sigma$ is defined by
$\tuple{(q,\vec{v}),(q',\vec{v}')}\in \tau$ $\diff$ 
$\exists t\in \Tau.\: q\rat{t}q' \wedge \vec{v}'\in t(\vec{v})$. 
The transfer functions in $\Tau$ include
imperative assignments and Boolean guards: 
if $b\in \wp(\bV^n)$ is a 
deterministic Boolean predicate (such as $x_1+2x_2 -1=0$) 
then the corresponding transfer function $t_b:\bV^n\ra \wp(\bV^n)$ 
is $t_b(\vec{v})\ud \textbf{if }\vec{v}\in b \textbf{ then } \{\vec{v}\}
\textbf{ else } \varnothing$.  
Examples of transfer functions include: affine, polynomial, nondeterministic assignments and
affine equalities guards.

The next value transformer $\post_{\tuple{q,q'}}:\wp(\bV^n)\ra \wp(\bV^n)$ for 
a pair $\tuple{q, q'}\in Q\times Q$ of control nodes is $\post_{\tuple{q,q'}}(X)\ud 
\cup \{t(X) \in \wp(\bV^n) \mid \exists t\in \Tau.q\rat{t}q'\}
$. 
Hence, if, for all $t\in \Tau$, 
$q \not\rat{t} q'$ then $\post_{\tuple{q,q'}}=\lambda X.\varnothing$.
The complete lattice
${\tuple{\wp(\Sigma),\subseteq}}$ of sets of states 
can be equivalently represented by 
the $Q$-indexed product lattice ${\tuple{\wp(\bV^n)^{|Q|},\dotted{\subseteq}}}$.
Hence, the successor transformer 
$\post_\cP:\wp(\bV^n)^{|Q|}\ra \wp(\bV^n)^{|Q|}$ 
and 
the set 
of reachable states of a program $\cP$ from  initial states in
$\Sigma_0\in \wp(\bV^n)^{|Q|}$
are defined as follows: 
\begin{align*}
&\post_\cP(\tuple{X_q}_{q\in Q}) \ud 
\tuple{\cup_{q\in Q} \post_{\tuple{q,q'}}(X_q)}_{q'\in Q}
\\ 
&\Reach[\cP,\Sigma_0]\ud \lfp(\lambda 
\vec{X}. \Sigma_0 \cup \post_\cP(\vec{X})) \in \wp(\bV^n)^{|Q|}
\end{align*}
For all control nodes $q\in Q$ and vector $\vec{X}\in \wp(\bV^n)^{|Q|}$, 
we will also use $\pi_q(\vec{X})$ and $\vec{X}_q$ to denote the $q$-indexed component of $\vec{X}$, \textit{e.g.}, $\Reach[\cP,\Sigma_0]_q\in \wp(\bV^n)$ will
be the set of reachable values at control node $q$.

We are interested in decidability and synthesis of abstract inductive 
invariants ranging in an abstract domain $A$ as specified
by a GI $(\wp(\bV^n)_\subseteq,\alpha,\gamma,A_{\leq_A})$ parametric on $n\in \bN$. 
By Corollary~\ref{coro-three}, for a given class $\cC$ of programs,  
a class $\Init$
of sets of initial states and a class $\Safe$ of sets of safety properties, 
the Monniaux problem~\eqref{ainv} 
is decidable iff for all $\cP=\tuple{Q,n,\bV,\Tau,\sra}\in \cC$,
$\Sigma_0\in \Init$ and $P\in \Safe$,  
\begin{equation}\label{kildall-q}
\!\!\!\!\!\dot{\gamma}(\lfp(\lambda \vec{a}\in A^{|Q|}\!. \:
\dot{\alpha}(\Sigma_0) \,\dot{\vee}_{A}\,  
\dot{\alpha} (\post_\cP (\dot{\gamma}(\vec{a}))))) \dotted{\subseteq}^? \!P 
\end{equation}
is decidable. Moreover, Corollary~\ref{coro-algo} provides 
an abstract inductive invariant
synthesis algorithm $\AInv$ for safety properties represented by $A$
(\textit{i.e.}, $P\in \gamma(A)$)
when 
$A$, $\cC$, $\Init$ and $\Safe$ satisfy the hypotheses of 
Corollary~\ref{coro-algo}.

In the following, we consider 
the decidability/synthesis  problem~\eqref{ainv} for Kildall's 
constant propagation  \cite{kildall73}
and Karr's affine equalities \cite{karr76} abstract domains.

\subsection{Kildall's Constant Propagation Domain}\label{const-sec}

Kildall's constant prop\-a\-ga\-tion \cite{kildall73} 
is a well-known and simple program analysis widely used in compiler optimization for
detecting whether a variable at some program point always stores a constant value for all possible program 
executions. Constant propagation relies on the abstract domain $\Const$, which, for simplicity, 
is here given for program variables assuming integer values in $\bZ$ ($\bQ$ and $\bR$ would be analogous).
For a single variable $\Const \ud \tuple{\bZ \cup \{\bot,\top\},\leq}$,
where $\leq$ is the standard flat partial order: 
for any $a\in \bZ \cup \{\bot,\top\}$, $\bot \leq a\leq \top$ (and $a\leq a$), 
which makes $\Const$ an infinite complete lattice with height~2. $\Const$ is straightforwardly specified by a GI $(\wp(\bZ)_\subseteq,\alpha_{\C},\gamma_{\C},\Const_\leq)$ where:
\begin{align*}
\alpha_{\C}(X)\ud 
\begin{cases}
\bot & \text{if }X=\varnothing\\
z & \text{if }X=\{z\}\\
\top & \text{otherwise}
\end{cases}
\qquad\qquad 
\gamma_{\C}(a) \ud 
\begin{cases}
\varnothing & \text{if }a=\bot\\
\{a\} & \text{if }a\in \bZ\\
\bZ & \text{if }a=\top 
\end{cases}
\end{align*}
For $n\geq 1$ variables, the 
product constant domain is
$\Const_n \ud (\Const_{\bot})^n \cup \{\bot\}$, where  
$\Const_{\bot}\ud \Const\smallsetminus\{\bot\}$, so as to have a unique
bottom element representing the empty set, 
while $\alpha_{\Const}:\wp(\bZ^n) \ra \Const$ and $\gamma_{\Const}:\Const\ra \wp(\bZ^n)$ 
are defined as pointwise lifting of, resp., $\alpha_{\C}$ and $\gamma_{\C}$.
$\tuple{\Const_n,\dot{\leq}}$ is a complete lattice of finite height $2n$, where 
lub and glb are defined pointwise. Finally, 
for a finite set of control nodes $Q$, $\dotted{\alpha}_{\Const}:
\wp(\bZ^n)^{|Q|} \ra \Const^{|Q|}$ and 
$\dotted{\gamma}_{\Const}:\Const^{|Q|}\ra \wp(\bZ^n)^{|Q|}$ 
are the $Q$-indexed pointwise lifting of, resp., $\alpha_{\Const}$ 
and $\gamma_{\Const}$.

It is known since \cite{hecht,reif} that the constant propagation problem 
is undecidable, meaning that for any program $\cP$ with $n$ variables, set $\Sigma_0$ of initial states,
node $q$ of $\cP$,  vector of constants (or $\top$) $\vec{k}\in \Const_n\smallsetminus\{\bot\}$, 
the problem $\alpha_{\Const}(\Reach[\cP,\Sigma_0]_q) \leq^? \vec{k}$ 
is undecidable.  
This undecidability is obtained by a simple reduction from
the undecidable Post correspondence problem 
and holds even if the interpretation of program branches is neglected, that is, in CFGs with unguarded
nondeterministic branching. 

By Corollary~\ref{coro-three}, 
the existence of abstract inductive 
invariants in $\Const$
for a given class $\cC$ of programs, 
$\Init$ of sets of initial states and $\Safe$ of sets of safety properties, 
is a decidable problem iff for all $\cP\in \cC$ with $n$ variables 
and control nodes in $Q$, 
for all $\Sigma_0\in \Init_{\cP}\subseteq \wp(\bV^n)^{|Q|}$ and for all $P\in \Safe_{\cP}\subseteq \wp(\bV^n)^{|Q|}$,  
\begin{equation}\label{kildall-q}
\!\!\!\!\!\dot{\gamma}_{\Const}(\lfp(\lambda \vec{a}\in \Const_n^{|Q|}\!. 
\dot{\alpha}_{\Const}(\Sigma_0) \dot{\vee}_{\Const_n}  \!
\dot{\alpha}_{\Const} (\post_\cP (\dot{\gamma}_{\Const}(\vec{a}))))) \dotted{\subseteq}^? \!P 
\end{equation}
is decidable.
We observe that:
\begin{enumerate}[{\rm (i)}]
\item if $P$ is a recursive set and $\vec{a}\in \Const_n^{|Q|}$ then 
$\dot{\gamma}_{\Const}(\vec{a}) \mathrel{\dot{\subseteq}^?} P$ is clearly decidable, 
because for all $q\in Q$, $i\in [1,n]$ and 
$k\in \bZ$, $k\in^? \pi_i(\pi_q(P))$ is decidable (and any $\dot{\gamma}_{\Const}(\vec{a}')$ is trivially recursive);  
\item if $\Sigma_0$ is a recursively enumerable (r.e.) set then 
$\dot{\alpha}_{\Const}(\Sigma_0)$ is clearly computable, because 
since $\Sigma_0$ is r.e., for all $q\in Q$ and $i\in [1,n]$, 
a (double) projection $\pi_i(\pi_q(\Sigma_0))$ is r.e., so that an algorithm enumerating
$\pi_i(\pi_q(\Sigma_0))$ allows us 
to compute its abstraction $\alpha_{\C}$ in $\Const$;
\item the binary lub $\vec{a} \vee_{\Const_n}  \vec{a}'$ in $\Const_n$ is clearly (pointwise) computable, thus, in turn, the $Q$-indexed lub 
$\dot{\vee}_{\Const_n}$ is computable.
\end{enumerate}

Therefore, since $\Const_n^{|Q|}$ has finite height,
a sufficient condition for the decidability of the problem
\eqref{kildall-q} is that the best correct approximation 
$\dot{\alpha}_{\Const} \circ \post_\cP \circ\: \dot{\gamma}_{\Const}: \Const_n^{|Q|} \ra 
\Const_n^{|Q|}$ is computable, where for all 
$\vec{a}\in \Const_n^{|Q|}$: %
\vspace*{-5pt}
\begin{align} \nonumber
\dot{\alpha}_{\Const} (\post_\cP(\dot{\gamma}_{\Const}(\vec{a}))) &= \nonumber\\
\tuple{\vee_{\Const}\{ \alpha_{\Const}(\post_{q,q'}(\pi_q(\dot{\gamma}_{\Const}(\vec{a})))) \mid q\in Q\}}_{q'\in Q} &= \nonumber\\
\tuple{\vee_{\Const}\{ \alpha_{\Const}(t(\pi_q(\dot{\gamma}_{\Const}(\vec{a})))) \mid \exists q\in Q, \exists t\in \Tau. q\rat{t} q'\}}_{q'\in Q}& \label{eq-const}
\end{align}
Because in \eqref{eq-const} we have a finite lub, it is enough that 
for all the transfer functions $t\in \Tau$ and $a\in\Const_n$, 
the bca $\alpha_{\Const}(t(\gamma_{\Const}(a)))$ is computable. 
It turns out that the best correct approximations in $\Const$ of 
affine assignments and linear Boolean guards are computable. Some simple
cases are described,
\textit{e.g.}, in \cite[Section~4.3]{mine17}, in the following we 
provide the algorithm of the bca's in full generality and detail.   

Let $\sum_{i=1}^n m_i x_i + b \in \LExp_n$ be a linear expression 
for $n$ integer variables ranging in $\Var_n\ud \{x_1,...,x_n\}$, where 
$\tuple{m_i}_{i=1}^n \in  \bZ^n$ is a vector of 
integer coefficients and $b\in \bZ$, and let $\grasse{\sum_{i=1}^n m_i x_i + b}:\wp(\bZ^n)\ra \wp(\bZ)$ denote its collecting semantics, that is,
$\grasse{\sum_{i=1}^n m_i x_i + b}X \ud 
\{\sum_{i=1}^n m_i v_i + b\in \bZ \mid \vec{v}\in X\}$. 
Then, we consider the transfer functions given by 
a single affine assignment
$t_a^j \dequiv x_j:=\sum_{i=1}^n m_i x_i + b$ and a linear Boolean guard 
$t_b \dequiv \sum_{i=1}^n m_i x_i + b \bowtie 0$ with $\bowtie \: \in \{=, \neq, <, \leq, >, \geq\}$, where for all $Y\in \wp(\bZ^n)$:
\begin{align*}
&t_a^j(Y) \ud  \{ \vec{y}[y_j/v_j] \in \bZ^n \mid \vec{y}\in Y, v_j \in 
\grasse{\textstyle{\sum}_{i=1}^n m_i x_i + b}\{\vec{y}\}\},\\
&t_b(Y) \ud \{ \vec{y} \in Y \mid 
\grasse{\textstyle{\sum}_{i=1}^n m_i x_i + b}\{\vec{y}\}\bowtie 0\}.
\end{align*}
These transfer functions can be easily extended to
include parallel affine assignments of the shape 
$\vec{x}:= M\vec{x} + \vec{b}$, where $M\in \bZ^{n\times n}$ is a 
$n\times n$ integer 
matrix and $\vec{b}\in \bZ^n$, which simply performs
$n$ parallel single affine assignments, 
and conjunctive (disjunctive) linear Boolean guards $ M\vec{x} + \vec{b} \bowtie 0$, 
which holds iff, for all (there exists) $j\in [1,n]$, 
$\sum_{i=1}^n M_{ji} x_i + b_j \bowtie 0$ holds. 

It turns out that the bca of $t_a^j$ and $t_b$ in $\Const_n$ 
are both computable. 
Let us first observe that if $\vec{a}\in \Const_n\smallsetminus \{\bot\}$ 
then $\grasse{\sum_{i=1}^n m_i x_i + b}(\gamma_{\Const}(\vec{a})) 
=\{k\}$ for some $k\in \bZ$
iff for all $l\in [1,n]$ either $m_l=0$ or  
$a_l\in \bZ$. Thus, for all $a\in \Const_n$, 
the problem $\exists k\in^{?} \bZ.\: \grasse{\sum_{i=1}^n m_i x_i + b}(\gamma_{\Const}(a)) 
=\{k\}$ is decidable, and in a positive case the constant $k$ is computable.
Moreover, if 
$\grasse{\sum_{i=1}^n m_i x_i + b}(\gamma_{\Const}(\vec{a})) 
\neq \{k\}$ then necessarily $\grasse{\sum_{i=1}^n m_i x_i + b}(\gamma_{\Const}(\vec{a}))=\bZ$.  
As a consequence of these observations, 
we derive the following definition:
for all $a\in \Const_n$,
\vspace*{-3pt}
\begin{align*}
\alpha_{\Const}(t_a^j(\gamma_{\Const}(a))) = 
\begin{cases}
\bot& \text{if } a=\bot\\
\vec{a}[a_j/k] & \text{if } a=\vec{a},\,
\grasse{\sum_{i=1}^n m_i x_i + b}(\gamma_{\Const}(\vec{a})) =\{k\}\\
\vec{a}[a_j/\top]       & \text{if } a=\vec{a},\,
\grasse{\sum_{i=1}^n m_i x_i + b}(\gamma_{\Const}(\vec{a})) =\bZ        
\end{cases}
\vspace*{-3pt}
\end{align*}
which shows that $\alpha_{\Const}(t_a^j(\gamma_{\Const}(a)))$ is computable and, in turn, bca's of parallel affine assignments are computable. 
The proof of computability of bca's of linear Boolean guards $t_b$ follows the same lines. 
We distinguish two cases: 
$\bowtie \;\in\! \{{\neq,}\: {<,}\: {\leq,}\: {>,}\: {\geq 
\}}$, whose transfer function is denoted by $t_b^{\bowtie}$, and
$\bowtie\;\in \{ = \}$, \textit{i.e.}\ affine equalities, 
denoted by $t_b^=$.
The corresponding algorithms are as follows: for all $a\in \Const_n$:
\vspace*{-3pt}
\begin{align*}
\alpha_{\Const}(t_b^{\bowtie}(\gamma_{\Const}(a))) =
\begin{cases}
\bot& \text{if } a=\bot\\
\vec{a} & \text{if } a=\vec{a},\, 
\grasse{\sum_{i=1}^n m_i x_i + b}(\gamma_{\Const}(\vec{a})) =\{k\}, k\bowtie 0\\
\bot & \text{if } a=\vec{a},\, 
\grasse{\sum_{i=1}^n m_i x_i + b}(\gamma_{\Const}(\vec{a})) =\{k\}, k\not \bowtie 0\\
\vec{a} & \text{if } a=\vec{a},\, 
\grasse{\sum_{i=1}^n m_i x_i + b}(\gamma_{\Const}(\vec{a})) =\bZ
\end{cases}
\end{align*}

\vspace*{-20pt}
\begin{multline*}
\alpha_{\Const}(t_b^=(\gamma_{\Const}(a))) =\\
\begin{cases}
\bot& \text{if } a=\bot\\
a & \text{if } a\neq \bot,\, 
\grasse{\sum_{i=1}^n m_i x_i + b}(\gamma_{\Const}(a)) =\{k\}, k= 0\\
\bot & \text{if } a\neq \bot,\, 
\grasse{\sum_{i=1}^n m_i x_i + b}(\gamma_{\Const}(a)) =\{k\}, k\neq 0\\
\vec{a}[a_j/k] & \text{if } a\neq \bot,\, 
\exists j\in[1,n]. m_j\neq 0,\, a_j=\top,\,  
\grasse{\sum_{i\neq j} m_i x_i + b}(\gamma_{\Const}(a)) =\{k\}\\
a & \text{otherwise}        
\end{cases}
\vspace*{-3pt}
\end{multline*}

Summing up, by Corollaries~\ref{coro-three} and \ref{coro-algo} 
we have shown the following result for the class
$\cC_{\Const}$ of 
nondeterministic programs with (possibly parallel) affine assignments
and (conjunctive or disjunctive) linear Boolean guards. 

\begin{theorem}[\textbf{Decidability and Synthesis of Inductive Invariants in $\Const$}]\label{deci-const}
The Monniaux problem \eqref{kildall-q} on $\Const$ for programs in
$\cC_{\Const}$,
r.e.\ sets of initial states and recursive sets of state
properties is decidable. Moreover, 
the algorithm $\AInv$ of Corollary~\ref{coro-algo} instantiated to
$\post_\cP$ for programs $\cP\in\cC_{\Const}$, r.e.\ sets of initial states and safety properties 
$P\in \dot{\gamma}_{\Const}(\Const^{|Q_{\cP}|})$ synthesizes the least
inductive invariant of $\cP$ in $\Const$, when this exists.
\end{theorem}

This result could be easily extended to polynomial
assignments and Boolean guards. We do not discuss the details of these
extensions since 
our main goal 
here was to illustrate on the simple example of 
$\Const$
how the abstract inductive invariant principle
can be applied to derive decidability  results and synthesis algorithms 
for abstract inductive invariants.

\begin{figure}[t]
\centering
\begin{tikzpicture}[shorten >=0pt,node distance=1.05cm,on grid,>=stealth',every state/.style={inner sep=0pt, minimum size=5mm, draw=blue!50,very thick,fill=blue!10}]
\node[state]         (q1)	{{\scriptsize $q_1$}}; 
\node[state]         (q2) [below=of q1] {{\scriptsize $q_2$}};
\node[state]         (q4) [below right=of q2] {{\scriptsize $q_4$}};
\node[state]         (q3) [below left=of q4] {{\scriptsize $q_3$}};

\path[->] (q1) edge node[align=right,right,midway] {{\scriptsize $x_1:=0;$}\\[-3pt] {\scriptsize $x_2:=2;$}} (q2)
          (q2) edge node[left,align=right,midway] {{\scriptsize $x_1:=x_1+2x_2;$}
          \\[-3pt] {\scriptsize $x_2:=x_2-1;$}} (q3)
          (q2) edge node[above=3pt,right] {{\scriptsize \;$x_1\geq 9$}} (q4);

\draw[->] (q3) .. controls +(left:3.2cm) and +(left:3.2cm) .. 
node[left,align=right,midway] {{\scriptsize $x_1:=x_1-x_2;$}\\[-3pt] {\scriptsize $x_2:=x_2+1;$}} (q2);
\end{tikzpicture}
\caption{A CFG representing a nondeterministic program $\cP$.}\label{fig2}
\end{figure}
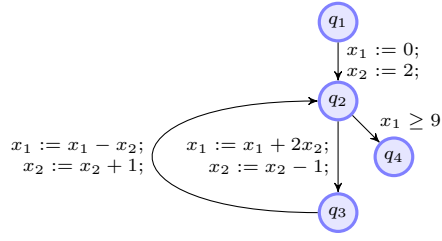

\begin{examplebf}\label{ex-one}
Let us consider the program $\cP$ represented in Fig.~\ref{fig2} with
$\Sigma_0={\{q_1\}\times \bZ^2}$  and the valid property 
$P^\sharp = \ok{\big\langle \tuple{q_1,(\top,\top)}, \!
\tuple{q_2,(\top,2)}, \! \tuple{q_3,(\top,\top)},$ $\tuple{q_4,(\top,\top)}\big\rangle} \in \ok{\Const_2^{|Q|}}$ representing that the variable $x_2$ is constantly equal to $2$ at the program point $q_2$.
It is easy to check that the algorithm $\AInv$ of Corollary~\ref{coro-algo}
yields the following sequence of abstract values $J^{k+1} =  
\dot{\alpha}_{\Const}(\post_{\cP}(\dot{\gamma}_{\Const}(J^k)))
\in \Const_2^{|Q|}$:
\begin{align*}
&J^0=\dot{\alpha}_{\Const}(\Sigma_0) = \big\langle\tuple{q_1,(\top,\top)}, 
\tuple{q_2,(\bot,\bot)},\tuple{q_3,(\bot,\bot)},\tuple{q_4,(\bot,\bot)}\big\rangle\\
&J^1=\big\langle\tuple{q_1,(\top,\top)}, 
\tuple{q_2,(0,2)},\tuple{q_3,(\bot,\bot)},\tuple{q_4,(\bot,\bot)}\big\rangle\\
&J^2=\big\langle\tuple{q_1,(\top,\top)}, 
\tuple{q_2,(0,2)},\tuple{q_3,(4,1)},\tuple{q_4,(\bot,\bot)}\big\rangle\\
&J^3=\big\langle\tuple{q_1,(\top,\top)}, 
\tuple{q_2,(\top,2)},\tuple{q_3,(4,1)},\tuple{q_4,(\bot,\bot)}\big\rangle\\
&J^4=\big\langle\tuple{q_1,(\top,\top)}, 
\tuple{q_2,(\top,2)},\tuple{q_3,(\top,1)},\tuple{q_4,(\top,2)}\big\rangle =\\ 
&\quad\;\! =
\dot{\alpha}_{\Const}(\post_{\cP}(\dot{\gamma}_{\Const}(J^4))) \dot{\leq}_{\Const} P^\sharp
\end{align*}
Thus, the output $J^4$ is the least (\textit{i.e.}\ strongest) 
inductive invariant in $\Const$ which allows us to prove $P^\sharp$.
\qed
\end{examplebf}

\paragraph{\textbf{Relationship with Completeness.}}
It is easy to check that  
all the transfer functions $t_a$ of affine assignments 
are pointwise complete, namely, $\alpha_{\Const} \circ t_a =
\alpha_{\Const} \circ t_a \circ \gamma_{\Const}\circ \alpha_{\Const}$ holds. 
Consequently, as recalled in Section~\ref{compl-sec}, from pointwise
completeness one obtains fixpoint completeness, that is, 
for all unguarded programs $\cP$ with 
(parallel) affine assignments, 
\[ \dot{\alpha}_{\Const}(\lfp(\lambda 
\vec{X}. \Sigma_0 \cup \post_\cP(\vec{X}))) = 
\lfp(\lambda 
\vec{a}. \dot{\alpha}_{\Const}(\Sigma_0) \:\dot{\sqcup}_{\Const}
\dot{\alpha}_{\Const}(\post_\cP(\dot{\gamma}_{\Const}(\vec{a})))).
\]
However,   fixpoint completeness is lost as soon
as affine linear Boolean guards are taken into account, 
because, in general, linear guards are not pointwise complete: for example, with 
$n=2$, we have that:
\begin{align*}
&\alpha_{\Const}(t_{x_1=0}(\{(1,0), (-1,0)\})) = \alpha_{\Const}(\varnothing)=
\bot \quad \text{while}
\\
&
\begin{aligned}
\alpha_{\Const}(t_{x_1=0}(\gamma_{\Const}(\alpha_{\Const}(\{(1,0), (-1,0)\}))) &= \alpha_{\Const}(t_{x_1=0}(\bZ\times \{0\}))\\
&=\alpha_{\Const}(\{(0,0)\})=\tuple{0,0}
\end{aligned}
\end{align*}
It is therefore worth remarking that the decidability result in $\Const$
for the programs in $\cC_{\Const}$ of Theorem~\ref{deci-const} holds even if
fixpoint completeness on $\Const$ for 
all the programs in $\cC_{\Const}$ does not hold.

\subsection{Karr's Affine Equalities Domain}

Program analysis on the domain 
of affine equalities has been introduced in 1976 by 
Karr~\cite{karr76}  who designed some algorithms
which compute for each program point the valid 
affine equalities between program variables.  
This abstract domain, here denoted by $\Aff$, is relatively simple, precise and is 
widely used in numerical program analysis (see, \textit{e.g.}, \cite{mine17,rival}). 
M\"{u}ller-Olm and Seidl \cite{seidl04}
put forward simpler and more efficient algorithms for the transfer functions of $\Aff$ and proved that   
Karr's analysis algorithm is fixpoint complete 
for unguarded nondeterministic affine programs, while 
for linearly guarded nondeterministic affine programs 
it is undecidable whether a 
given affine equality holds at a program point or not.

Let us briefly recall the definition of the abstract domain $\Aff$, where here the $n$
program variables assume rational values, that is, $\bV=\bQ$. 
The abstract invariants in $\Aff$ are finite (possibly empty) conjunctions
of affine equalities between variables in $\Var_n$, namely, $\bigwedge_{j=1}^k (\sum_{i=1}^n m_{i,j} x_i + b_j =0)$.
The geometric view of $\Aff$ is based on 
the affine hull of a subset $X\in \wp(\bQ^n)$ which is defined by: 
\begin{equation*}
\aff(X)\ud \{\textstyle{\sum_{j=0}^m} \lambda_j \vec{x}_j \in \bQ^n \mid m\in \bN,\, \lambda_j \in \bQ,\, 
\vec{x}_j\in X,\, \textstyle{\sum_{j=0}^m} \lambda_j=1 \}.
\end{equation*}
This map $\aff:\wp(\bQ^n) \ra \wp(\bQ^n)$ is an upper 
closure  on $\tuple{\wp(\bQ^n),\subseteq}$ whose
fixpoints are the affine subspaces in $\bQ^n$ and 
therefore defines the affine equalities domain $\Aff\ud \tuple{\aff(\wp(\bQ^n)),\subseteq}$.
Thus, any conjunction of affine equalities 
can be represented by an affine subspace, and vice versa. 
$\tuple{\aff(\wp(\bQ^n)),\subseteq}$ is closed under arbitrary intersections, because $\aff$ is a uco, but not under unions, \textit{i.e.}, $\aff$ is not
an additive uco. 
The lub in $\Aff$ of a set of affine subspaces $\cX\subseteq \Aff$ is given by 
$\sqcup_{\Aff} \cX \ud \aff(\cup_{X \in \cX} X)$. 
A complex 
algorithm for computing a binary lub $A \sqcup_{\Aff} A'$
of two affine subspaces $A,A'\in \Aff$ was given by 
Karr~\cite[Section~5.2]{karr76}, 
a simpler and more efficient version
is described by Min\'e \cite[Section 5.2.2]{mine17}.
$\tuple{\Aff,\subseteq}$ is a complete lattice of finite height $n+1$, because
if $A,A'\in \Aff$ and $A\subsetneq A'$ then $\dim(A) < \dim(A')$, 
where $\dim(\varnothing)=-1$ and $\dim(\bQ^n)=n$. 
\\
\indent
Karr gave already in 
\cite[Section 4.2]{karr76} an algorithm for computing the bca of 
an affine assignment 
$t_a^j \equiv x_j:=\sum_{i=1}^n m_i x_i + b$, with $j\in [1,n]$.  
Karr's algorithm is 
based on representing affine subspaces by kernels
of affine transformations (\textit{i.e.},  matrices) 
and distinguishes between invertible (such as 
$x_i := x_i +k$) and noninvertible (such as $x_i := x_j +k$)
affine assignments.  
M\"{u}ller-Olm and Seidl \cite{seidl04} put forward a
more efficient algorithm  which is based on representing affine
subspaces by affine bases. It is worth remarking that \cite[Lemma~2]{seidl04} also observes
that the bca of $t_a^j$ is pointwise complete, namely 
$\aff \circ \: t_a^j \circ \aff  = \aff \circ\: t_a^j$ holds. 
We are not interested here in the details and complexities 
of these algorithms which implement
the best abstraction in $\Aff$ of the affine assignments $t_a^j$, 
since here we just exploit the fact that the best correct approximation
$\aff \circ\: t_a^j : \Aff \ra \Aff$ is computable. 
In turn, computability of parallel affine assignments 
$\vec{x}:= M\vec{x} + \vec{b}$ easily follows. 
\\
\indent
M\"{u}ller-Olm and Seidl \cite{seidl04} also show that the bca of 
a nondeterministic assignments  $\ok{t^j_{a?}\dequiv x_j :=?}$ 
is computable, where for all $Y\in \wp(\bQ^n)$, the corresponding transfer function is defined by:
$t_{x_j:=?}(Y) \ud \{\vec{y}[y_j/v] \in \bQ^n \mid \vec{y}\in Y,\, 
v \in \bQ\}$.
An observation in \cite[Lemma~4]{seidl04} 
states that $\aff (t_{x_j:=?}(\aff(Y))
= \aff (t_{x_j:=0}(\aff(Y)) \sqcup_{\Aff} \aff (t_{x_j:=1}(\aff(Y))$, thus reducing the computation of 
$\aff (t_{x_j:=?}(\aff(Y))$ to computing the lub
of the bca's of the transfer functions 
of two single affine assignments
$x_j:=0$ and $x_j:=1$.
\\
\indent
As already shown by Karr \cite[Section~4.1]{karr76} 
(see also \cite[Section~5.2.3]{mine17} for a modern approach),
it turns out that best correct approximations of 
affine equalities Boolean guards of the shape 
$\ok{t_b^= \dequiv \sum_{i=1}^n m_i x_i + b = 0}$ are computable.
 However, 
no algorithm for computing the bca of a generic negation 
$\ok{t_b^{\neq} \dequiv \sum_{i=1}^n m_i x_i + b \neq 0}$ is known 
in literature. Both 
Karr~\cite[Section~4.1]{karr76} and Min\'e \cite[Section~5.2.3]{mine17}
propose to soundly approximate a negated affine equality guard $t_b^{\neq}$ simply by the identity
transfer function, meaning that 
$t_b^{\neq}$ is replaced by the $\mathit{true}$ predicate. 
Karr~\cite[Section~4.1]{karr76} mentions  that
``a general study 
of how best to handle decision nodes which are not of the simple form $t_b^=$ is \emph{in preparation}'', but this document never appeared. To the best
of our knowledge, the literature provides no algorithm for 
computing the bca of $t_b^{\neq}$ in $\Aff$, and we conjecture that, in general,
it could not be computable.

Summing up, as a consequence of Corollaries~\ref{coro-three} and~\ref{coro-algo} 
we therefore obtain the following result for the class $\cC_{\Aff}$ of 
nondeterministic programs with (possibly parallel) affine assignments, 
(possibly parallel) nondeterministic assignments and 
(conjunctive or disjunctive) positive linear equalities guards.

\begin{theorem}[\textbf{Decidability and Synthesis of Inductive Invariants in $\Aff$}]\label{th-karr}
The Monniaux problem \eqref{kildall-q} on $\Aff$ for programs in $\cC_{\Aff}$,
r.e.\ sets of initial states and recursive sets of state
properties is decidable. 
Moreover, 
the algorithm $\AInv$ of Corollary~\ref{coro-algo} instantiated to
$\post_\cP$ for programs $\cP\in \cC_{\Aff}$, r.e.\ sets of initial states and safety properties 
$P\in \Aff$ synthesizes the least
inductive invariant of $\cP$ in $\Aff$, when this exists. 
\end{theorem}

\begin{figure}[t]
\centering
\qquad\begin{tikzpicture}[shorten >=0pt,node distance=1.2cm,on grid,>=stealth',every state/.style={inner sep=0pt, minimum size=5mm, draw=blue!50,very thick,fill=blue!10}]
\node[state]         (q1)	{{\scriptsize $q_1$}}; 
\node[state]         (q2) [below=of q1] {{\scriptsize $q_2$}};
\node[state]         (q3) [below=of q2] {{\scriptsize $q_3$}};
\node[state]         (q4) [below=of q3] {{\scriptsize $q_4$}};

\path[->] (q1) edge [right,align=left,midway] node {{\scriptsize \;\,$x_1:=-2;$}\\[-3pt] {\scriptsize \;\,$x_2:=1;$} \\[-3pt] {\scriptsize \;\,$x_3:=1;$}} (q2)
          (q2) edge [bend right=90,out=290,in=240] node[left,align=left,midway] {{\scriptsize$x_1:=-2x_2- 2;$} \\[-3pt] {\scriptsize  $x_2:=x_2 +x_3;$}} (q3)
          (q2) edge [bend left=90,out=70,in=120] node[right,align=left,midway] {{\scriptsize$x_1:=2x_1 +4;$}  \\[-3pt] {\scriptsize $x_2:=-x_1-2x_3;$}} (q3)
          (q3) edge [right,align=left,midway] node {{\scriptsize $x_1+2x_3=0$}} (q4)
          ;

\draw[->] (q3) .. controls +(right:4.3cm) and +(right:4.3cm) .. 
node[right] {}   (q2);
\end{tikzpicture}
\caption{A nondeterministic program $\cR$.}\label{fig-cfg}
\end{figure}

\begin{examplebf}
Let us consider the nondeterministic program $\cR$ in Fig.~\ref{fig-cfg} with
$\Sigma_0=\{q_1\}\times \bQ^3$ and the property 
$P^\sharp = \ok{\big\langle \tuple{q_1,\top}, \tuple{q_2,\top},\tuple{q_3,\top},\tuple{q_4, x_1 +x_2 +1 =0} \big\rangle}\in \ok{\Aff^{|Q|}}$. 
The algorithm $\AInv$ of Corollary~\ref{coro-algo}
yields the following sequence of abstract values $I^j\in \ok{\Aff^{|Q|}}$:
\begin{align*}
&I^0=\dot{\alpha}_{\Aff}(\Sigma_0)=\big\langle\tuple{q_1,\top}, 
\tuple{q_2,\bot},\tuple{q_3,\bot},\tuple{q_4,\bot}\big\rangle\\
&I^1=\big\langle\tuple{q_1,\top}, 
\tuple{q_2,x_1+2=0 \wedge x_2-1=0 \wedge x_3-1=0},\tuple{q_3,\bot},\tuple{q_4,\bot}\big\rangle\\
&I^2=\big\langle\tuple{q_1,\top}, 
\tuple{q_2,x_1+2=0 \wedge x_2-1=0 \wedge x_3-1=0},\\
&\qquad\qquad\qquad\, \tuple{q_3,x_1+2x_2=0 \wedge x_3=1},\tuple{q_4,\bot}\big\rangle\\
&I^3=\big\langle\tuple{q_1,\top}, 
\tuple{q_2,x_1+2x_2=0 \wedge x_3=1},\tuple{q_3,x_1+2x_2=0 \wedge x_3=1},\\&\qquad\qquad\qquad\,\tuple{q_4,x_1+x_2+1=0 \wedge x_3=1,\bot}\big\rangle=
\dot{\alpha}_{\Aff}(\post_{\cR}(\dot{\gamma}_{\Aff}(I^3)))\:\dot{\leq}_{\Aff}\, P^\sharp
\end{align*}
The output $I^3$ is the analysis of $\cR$ with the bca's of its transfer functions in $\Aff$, \textit{i.e.},  it is 
the least 
inductive invariant in $\Aff$ which allows us to prove that $P^\sharp$ holds.
\qed
\end{examplebf}

\paragraph{\textbf{Relationship with Completeness.}}
M\"{u}ller-Olm and Seidl~\cite{seidl04} implicitly show that 
the transfer functions of affine assignments $t_a$ and  
of nondeterministic assignments $t_{x_j:=?}$ are pointwise complete. 
In fact, \cite[Lemma~2]{seidl04} shows that for 
all $X\in \wp(\bQ^n)$, $t_a(\aff(X))=\aff(t_a(X))$, from which we easily obtain: 
\begin{align*}
\aff(t_a(X)) &=\aff(\aff(t_a(X)))=\aff(t_a(\aff(X))) 
\\[4pt]
\aff (t_{x_j:=?}(X)) &= \aff(\cup_{z\in \bQ} t_{x_j:=z}(X)) = \aff(\cup_{z\in \bQ} \aff(t_{x_j:=z}(X)))=\\
& \hspace*{-55pt}= \aff(\cup_{z\in \bQ} \aff(t_{x_j:=z}(\aff(X))) = \aff(\cup_{z\in \bQ} t_{x_j:=z}(\aff(X)) = \aff(t_{x_j:=?}(\aff(X)))
\end{align*}
Thus, since pointwise
completeness entails fixpoint completeness, similarly to 
Section~\ref{const-sec},
for all unguarded programs $\cP$ with 
affine and nondeterministic assignments, 
$\dot{\alpha}_{\Aff}(\lfp(\lambda 
\vec{X}. \Sigma_0 \cup \post_\cP(\vec{X}))) = 
\lfp(\lambda 
\vec{a}. \dot{\alpha}_{\Aff}(\Sigma_0)\: \dot{\sqcup}_{\Aff}$ $\dot{\alpha}_{\Aff}(\post_\cP(\dot{\gamma}_{\Aff}(\vec{a}))))$ holds. 
Here, fixpoint completeness is lost as soon
as affine equality guards are included (these programs still belong
to $\cC_{\Aff}$), because they are not pointwise complete, \textit{e.g.}: 
\begin{align*}
&\aff(t_{x_1=0}(\{(1,0), (-1,0)\})) = \aff(\varnothing)=
\varnothing \quad \text{while}\\
&\aff(t_{x_1=0}(\aff(\{(1,0), (-1,0)\})))=
\aff(t_{x_1=0}(x_2=0)) = \aff(\{(0,0)\})= \{(0,0)\}
\end{align*}
Moreover, a result in \cite[Section~7]{seidl04}
proves that
once affine equality guards are added to 
nondeterministic affine programs 
it becomes undecidable whether a 
given affine equality holds in some program point or not.
This undecidability 
does not prevent the decidability result of Theorem~\ref{th-karr}, 
since these two decision problems are orthogonal. 

\section{Co-Inductive Synthesis of Abstract Inductive Invariants}

In the following we design a synthesis algorithm which, 
by generalizing a recent algorithm by Padon et al.~\cite{padon16}, 
outputs the \emph{most abstract} inductive invariant in an
abstract domain, when this exists. This algorithm is obtained by dualizing 
the procedure $\AInv$ in Corollary~\ref{coro-algo} to a co-inductive greatest fixpoint computation 
and will require that 
the abstract domain is equipped 
with a suitable well-quasiorder relation.
Let us recall that a qoset $D_\leq$ is a well-quasiordered set (wqoset), and $\leq$ is called well-quasiorder (wqo), when for 
every countably infinite  sequence of elements \(\{x_i\}_{i\in \bN}\) in $D$ there exist \(i,j\in \bN\) such that \(i<j\) and \(x_i\leq x_j\).
Equivalently,  \(D\) is a wqoset iff $D$ is DCC (also called well-founded) and 
$D$  has no infinite antichain.

Let $\cT=\tuple{\Sigma,\tau}$ 
be a transition system whose successor transformer is $\post$, 
so that $\lfp(\lambda X.\Sigma_0 \cup \post(X))\in \wp(\Sigma)$ 
are the reachable states of $\cT$ from some set of initial states $\Sigma_0\in \wp(\Sigma)$. \cite{padon16} considers
abstract invariants ranging in a set (of semantics of logical formulae)
$L\subseteq \wp(\Sigma)$ and assumes (in \cite[Theorem~4.2]{padon16}) that $\tuple{L,\subseteq}$ is closed under finite intersections (\textit{i.e.}, logical conjunctions).
Accordingly to Assumption~\ref{assu}, 
we ask that $\tuple{L,\subseteq}$ satisfies the 
requirement  of being an abstract domain of the concrete
domain $\tuple{\wp(\Sigma),\subseteq}$, which corresponds to ask that $\tuple{L,\subseteq}$  is closed 
under arbitrary, rather than finite, intersections. Thus,
$L$ is the image of an upper closure 
$\umul\in \uco(\wp(\Sigma)_\subseteq)$ defined by: 
$\umul(X)\ud 
\cap \{\phi \in L \mid X \subseteq \phi\}$. 

The three key definitions and related assumptions of 
the synthesis algorithm defined in
\cite[Theorem~4.2]{padon16} concern a
quasi-order $\sqsubseteq_L \: \mathrel{\subseteq} \Sigma\times \Sigma$ between states, 
a function $\av_L:\Sigma \ra \wp(\Sigma)$ called Avoid, and an abstract transition relation $\tau^L \subseteq \Sigma \times \Sigma$:
\begin{align*}
&{\rm (1)\;} s \sqsubseteq_L s' \;\diff\; \forall \phi \in L. s'\in \phi \Ra s \in \phi&& \text{Assumption~(A$_1$): $\tuple{\Sigma,\sqsubseteq_L}$ is a wqoset}  \\
&{\rm (2)\,} \av_L(s)\ud \cup \{\phi \in L \mid \phi \subseteq \neg \{s\}\}
&& \text{Assumption~(A$_2$): $\forall s\in \Sigma.\:\av_L(s)\in L$}\\
&{\rm (3)\;} (s,s')\in \tau^L \;\diff\; (s,s')\in \tau \vee s' \sqsubseteq_L s&&
\text{Abstract Transition System $\cT^L\ud \tuple{\Sigma,\tau^L}$}
\end{align*}

Correspondingly, we define the down-closure $\delta_L:\wp(\Sigma)\ra \wp(\Sigma)$ 
of $\sqsubseteq_L$, 
we generalize $\av_L:\wp(\Sigma) \ra \wp(\Sigma)$ to sets
of states and we define the successor transformer 
$\post^L: \wp(\Sigma)\ra \wp(\Sigma)$ of $\cT^L$ as follows:
\begin{align*}
&{\rm (1)\;} \delta_L(X)\ud \{s\in \Sigma \mid \exists s'\in X. s \sqsubseteq_L s'\} &&\text{Down-closure of the qo $\sqsubseteq_L$}\\
&{\rm (2)\;} \av_L(X)\ud 
\cup \{\phi \in L \mid \phi \subseteq \neg X\} &&\text{$\av_L$ for any set $X$ of states}\\
&{\rm (3)\;} \post^L(X) \ud \post(X) \cup \delta_L(X) &&\text{Successor transformer of $\cT^L$} 
\end{align*}

\begin{lemma}\label{lemma6} The following properties hold: 
\begin{enumerate}[{\rm (a)}]
\item $s \sqsubseteq_L s'$ iff $\umul(\{s\})\subseteq \umul(\{s'\})$.

\item {\rm (A$_1$)} holds iff 
$\tuple{L,\subseteq}$ is a well-quasi order.
\item {\rm (A$_2$)} holds iff $L$ is closed under arbitrary unions.
\item If {\rm (A$_2$)} holds then 
$\delta_L = \umul$ and both are additive ucos (and $\delta_L(\phi)=\phi \Lra \phi\in L$).%
\item For all $\Sigma_0\in \wp(\Sigma)$,
$\lfp(\lambda X. \Sigma_0 \cup \post^L(X))=
\lfp(\lambda X.\umul(\Sigma_0 \cup \post(X)))$.
\end{enumerate}
\end{lemma}
\begin{proof} 
(a) If  $s \sqsubseteq_L s'$ and $t\in \umul(\{s\})$ then $t\in \umul(\{s'\})=\cap 
\{\phi \in L \mid s'\in \phi\}$: if $\phi \in L$ and $s'\in \phi$ then $s\in \phi$, so 
that, since $t\in \umul(\{s\})$, $t\in \phi$. Conversely, if $\umul(\{s\})\subseteq \umul(\{s'\})$,
$\phi \in L$ and $s'\in \phi$, then, since $s\in \umul(\{s'\})$, $s\in \phi$.

\noindent
(b) By (a), we equivalently prove that
$\tuple{\{\umul(\{s\}) \mid s\in \Sigma\},\subseteq}$ is a wqo iff 
$\tuple{L,\subseteq}$ is a wqo.
 
\noindent 
$(\Ra)$: \cite[Lemma~4.6]{padon16} proves $\tuple{L,\subseteq}$ is well-founded, 
we additionally show that it does not contain infinite antichains. 
By contradiction, assume that $\{\phi_i\}_{i\in \bN}$ is an infinite antichain
in $\tuple{L,\subseteq}$. Thus, for all $i\neq j$, $\phi_i\not\subseteq\phi_j$
and $\phi_j\not\subseteq \phi_i$, so that there exist $s_{i,j}\in \phi_i\smallsetminus \phi_j$ and $s_{j,i}\in \phi_j\smallsetminus \phi_i$.
From $s_{j,i}\in \phi_j$ we obtain that 
$\umul(\{s_{j,i}\})\subseteq \umul(\phi_j)=\phi_j$. From $s_{i,j}\not\in
\phi_j$, we obtain that $s_{i,j}\in \umul(\{s_{i,j}\})\not\subseteq \phi_{j,i}$. 
It turns out that $\umul(\{s_{i,j}\})\not\subseteq \umul(\{s_{j,i}\})$,
otherwise from $s_{i,j}\in \umul(\{s_{i,j}\})
\subseteq \umul(\{s_{j,i}\})\subseteq \phi_j$ we would obtain
the contradiction $s_{i,j}\in\phi_j$. Dually, $\umul(\{s_{j,i}\})\not\subseteq \umul(\{s_{i,j}\})$ holds. 
Thus, for any $i\in \bN$, $\{\umul(\{s_{i,j}\}) \mid j\in \bN, j\neq i\}$ is an infinite antichain in 
$\tuple{\{\umul(\{s\}) \mid s\in \Sigma\},\subseteq}$,
which is a contradiction.

\noindent
$(\La)$: $\tuple{\{\umul(\{s\}) \mid s\in \Sigma\},\subseteq}$ is trivially 
a wqo because $\{\umul(\{s\}) \mid s\in \Sigma\}\subseteq L$ and $L$ is a wqo.

\noindent
(c)
Assume that for all $s\in \Sigma$, 
$\av_L(s)\in L$.
Let us show that for all $S\in \wp(\Sigma)$, $\cap_{s\in S} \av_L(s) = \av_L(S)$. 

\noindent 
$(\supseteq)$: Let $t\in \phi$ for some $\phi\in L$ such that $\phi \subseteq \neg S$. Then, for all $s\in S$, $\phi \subseteq \av_L(s)$, so that $t\in 
\cap_{s\in S} \av_L(s)$.

\noindent 
$(\subseteq)$:
Let $t\in \cap_{s\in S} \av_L(s)$. For all $s\in S$, there exists $\phi_s\in L$
such that $\phi_s \subseteq \neg \{s\}$ and $t\in \phi_s$. Thus, 
$\cap_{s\in S}\phi_s\in L$ and $t\in \cap_{s\in S}\phi_s \subseteq \neg S$, meaning that $t\in \av_L(S)$. 

\noindent
Thus, since $L$ is assumed to be closed under arbitrary intersections we obtain that $\av_L(S)=\cap_{s\in S} \av_L(s)\in L$. 
Consider now $\Phi \subseteq L$.
Then, $\av_L(\neg (\cup \Phi)) = \cup\{\phi \in L \mid \phi \subseteq 
\neg\neg (\cup \Phi)\}=\cup\{\phi \in L \mid \phi \subseteq 
\cup \Phi\} = \cup\Phi$, so that $\cup \Phi\in L$.

\noindent
Conversely, if $L$ is closed under arbitrary unions then $\av_L(s)=
\cup \{\phi \in L \mid \phi \subseteq \neg \{s\}\}\in L$.

\noindent
(d) Since $\sqsubseteq_L$ is a quasi-order relation, its down-closure $\delta_L$ is an upper closure on $\tuple{\wp(\Sigma),\subseteq}$. 
By (a), $\delta_L(X)= \{s\in \Sigma \mid \exists s'\in X. s \sqsubseteq_L s'\}=
\{s\in \Sigma \mid \exists s'\in X. \umul(\{s\}) \subseteq \umul(\{s'\})\} =
\{s\in \Sigma \mid \exists s'\in X. s\in \umul(\{s'\})\}
=\cup_{s\in X} \umul(\{s\})$. By (c), since $L$ is closed under arbitrary unions, the upper closure $\umul$ is additive, so
that $\cup_{s\in X} \umul(\{s\})= \umul(\cup_{s\in X}\{s\})=\umul(X)$, consequently $\delta_L(X)=\umul(X)$. 
In particular, $\delta_L(\phi)=\phi \Lra \umul(\phi)=\phi \Lra \phi\in L$.

\noindent
(e)
By Lemma~\ref{lemma6}~(d), $\lfp(\lambda X. \post(X)\cup \delta_L(X)\cup \Sigma_0)=
\lfp(\lambda X. \post(X)\cup \umul(X)\cup \Sigma_0)$.
It turns out that
\begin{align*}
\Sigma_0  \cup \post(X)\cup \delta_L(X) \subseteq X &\Lra 
\quad\text{[by Lemma~\ref{lemma6}~(d)]}\\
\Sigma_0  \cup \post(X)\cup \umul(X)\subseteq X&\Lra 
\quad\text{[by set theory]}\\
\Sigma_0  \subseteq X \wedge \post(X)\subseteq X \wedge \umul(X) \subseteq X  &\Lra
\quad\text{[as $\umul$ is a uco]}\\
\Sigma_0 \subseteq X  \wedge \post(X)\subseteq X \wedge \umul(X)= X &\Lra
\quad\text{[as $\umul(X)=X$]}\\
\Sigma_0 \subseteq X \wedge \post(\umul(X))\subseteq X\wedge \umul(X)=X 
 &\Lra
\quad\text{[by set theory]}\\
\Sigma_0 \cup \post(\umul(X))\cup \umul(X) \subseteq X &
\end{align*}
Thus,  by Knaster-Tarski theorem, 
$\lfp(\lambda X.\Sigma_0 \cup \post(\umul(X)) \cup \umul(X))
=\lfp(\lambda X. \Sigma_0 \cup \post(X)\cup \delta_L(X))$. 
We also have that:
\begin{align*}
\Sigma_0 \cup \post(\umul(X))\cup \umul(X) \subseteq X &\Lra\quad\text{[by the equivalences above]}\\
\Sigma_0 \cup \post(\umul(X))\cup \umul(X) \subseteq X=\umul(X) &\Lra\quad\text{[by set theory]}\\
\Sigma_0 \cup \post(\umul(X)) \subseteq X=\umul(X) &\Lra\quad\text{[as $\umul(X)=X$]}\\
\Sigma_0 \cup \post(X) \subseteq X=\umul(X) &\Lra\quad\text{[as $\umul$ is a uco]}\\
\umul(\Sigma_0 \cup \post(X)) \subseteq X=\umul(X) &\Lra\quad\text{[by set theory]}\\
\umul(\Sigma_0 \cup \post(X)) \subseteq X &
\end{align*}
Thus, by Knaster-Tarski theorem, $\lfp(\umul(\Sigma_0 \cup \post(X))) \subseteq
\lfp(\lambda X.\Sigma_0 \cup \post(\umul(X)) \cup \umul(X))$ follows. 
Moreover, if $F\ud \lfp(\umul(\Sigma_0 \cup \post(X)))$, so that
$F=\umul(F) = \Sigma_0 \cup \post(F)$, then
$\Sigma_0 \cup \post(\umul(F)) \cup \umul(F)= \Sigma_0 \cup \post(F) \cup \umul(F) = F$, and this implies that 
$\lfp(\lambda X.\Sigma_0 \cup \post(\umul(X)) \cup \umul(X))
\subseteq \lfp(\umul(\Sigma_0 \cup \post(X)))$. Therefore, 
$\lfp(\lambda X.\Sigma_0 \cup \post(\umul(X)) \cup \umul(X))
=\lfp(\umul(\Sigma_0 \cup \post(X)))$.
\qed\end{proof}

In particular, Lemma~\ref{lemma6}~(e) states that the set 
of reachable states in the abstract transition system 
$\cT^L$ coincides with the set of reachable states
of the abstract transition system $\cT^{\umul}\ud \tuple{\Sigma,\umul \circ \post_\cT}$ obtained by considering the
best correct approximation of $\post_\cT$ in the abstract domain $\umul$.  As a direct consequence of the abstract inductive invariant principle 
Lemma~\ref{char}~(a), we 
obtain the following characterization of the abstract 
inductive invariants ranging in $L$ of the transition system $\cT$ 
 which relies on the 
reachable states of  its 
best abstraction  $\cT^\umul$ in the domain $\umul$.

\begin{corollary}\label{coro9}
Let $\Sigma_0,P\in \wp(\Sigma)$. Then,
$\exists \phi \in L. \Sigma_0 \subseteq \phi\wedge \post(\phi)\subseteq \phi \wedge \phi \subseteq P$ iff $\Reach[\cT^{\umul},\Sigma_0]\subseteq P$. 
\end{corollary}
\begin{proof}
It turns out that
\begin{align*}
\lfp(\lambda X.\umul(\Sigma_0 \cup \post(X) )) \subseteq P&\Lra
\quad\text{[by Lemma~\ref{char}~(a) for ucos]}\\
\exists \phi \in \umul(\wp(\Sigma)). \Sigma_0 \cup \post(\phi) \subseteq \phi\wedge \phi \subseteq P &\Lra 
\quad\text{[as $\umul(\wp(\Sigma)) = L$]}\\
\exists \phi \in L. \Sigma_0 \subseteq \phi\wedge \post(\phi)\subseteq \phi \wedge \phi \subseteq P& \tag*{\qed}
\end{align*}
\end{proof}

\subsection{Co-Inductive Invariants}
Following \cite{padon16}, in the following we make assumption (A$_2$), 
that is, by Lemma~\ref{lemma6}~(c), 
we assume that $L\subseteq \wp(\Sigma)$ 
is closed under arbitrary unions. This means that 
$\umul$ is an additive
uco on $\wp(\Sigma)_\subseteq$, \textit{i.e.}, in abstract interpretation terminology, $\umul$ is a \emph{disjunctive} abstract domain  whose abstract lub does
not lose precision (see, \textit{e.g.},\cite[Section~6.3]{mine17}). 
Furthermore, we also have that 
$L$ is the image of a co-additive (\textit{i.e.}, preserving arbitrary intersections)
lower closure $\lmul:\wp(\Sigma)\ra \wp(\Sigma)$
defined by $\lmul(X)\ud \cup\{\phi\in L\mid \phi\subseteq X\}$.
It turns out that the 
uco $\umul$ is adjoint to the  lco $\lmul$ (this is 
called abstract adjoinedness in \cite[Section~3.5]{cou00}),  namely,  
$\umul(X)\subseteq Y \Lra X \subseteq \lmul(Y)$. 
In fact, if $\umul(X)\subseteq Y$ then,
by applying $\lmul$, $X\subseteq \umul(X)=\lmul(\umul(X))\subseteq \lmul(Y)$; the converse is dual.

As observed in \cite[Theorem~4]{cou00}, the inductive invariant principle~\eqref{inv} 
can be dualized
when $f$ admits right-adjoint  $\ok{\widetilde{f}}:C\ra C$ 
(this happens iff $f$ is additive): in this case, 
\begin{equation}\label{dualization}
\lfp(\lambda x. c\vee f(x))\leq c' \Lra c\leq \gfp(\lambda x.c'\wedge \ok{\widetilde{f}}(x))
\end{equation}
holds and one obtains a \emph{co-inductive} invariant principle:
\begin{equation}\label{dinv}
c\leq \gfp(\lambda x.\ok{\widetilde{f}}(x)\wedge c') \Lra \exists j\in C.\: c\leq j \wedge 
j\leq \ok{\widetilde{f}}(j) \wedge j\leq c'
\end{equation} 
One such
$j\in C$ is therefore called a \emph{co-inductive invariant} of $f$
for $\tuple{c,c'}$. 

The co-inductive invariant 
proof method~\eqref{dinv} can be applied to safety verification of any transition system $\cT$ because $\post$ is additive and therefore it always admits right adjoint $\pret$ (cf.\ Section~\ref{sec:background}). Hence, we obtain  that
$\lfp(\lambda 
X. \Sigma_0 \cup \post(X))\subseteq P$ iff $\Sigma_0 \subseteq \gfp(\lambda 
X. \pret(X)\cap P)$ iff there exists a co-inductive invariant for $\pret$
for $\tuple{\Sigma_0,P}$.
By \eqref{inv} and \eqref{dinv}, it turns out that $I$ is an inductive invariant of $\post$ for $\tuple{\Sigma_0,P}$ 
iff $I$ is a co-inductive
invariant of $\pret$ for $\tuple{\Sigma_0,P}$. Also, while
$\lfp(\lambda X. \Sigma_0 \cup \post(X))$ is the least, \textit{i.e.}\ logically strongest, inductive invariant, we have that
$\gfp(\lambda X. \pret(X)\cap P)$ is the greatest, \textit{i.e.}\ logically weakest,  inductive invariant
\cite[Theorem~6]{cou00}.

We show how the co-inductive invariant principle \eqref{dinv}
applied 
to the best abstract transition system 
$\cT^{\umul}= (\Sigma,\umul \circ \post_\cT)$ provides exactly the 
synthesis algorithm by Padon et al.~\cite[Algorithm~1]{padon16}. 
In order to do this, we first observe the following alternative characterization 
of the reachable states of $\cT^\umul$.

\begin{lemma}\label{lemma9}
$\lfp(\lambda X. \Sigma_0 \cup \umul(\post(X)))=
\lfp(\lambda X.\Sigma_0 \cup \post(\umul(X)) \cup \umul(X))
$.
\end{lemma}
\begin{proof} 
The proof of Lemma~\ref{lemma6}~(e) shows that
$\lfp(\lambda X.\Sigma_0 \cup \post(\umul(X)) \cup \umul(X))
=\lfp(\umul(\Sigma_0 \cup \post(X)))$.
Moreover, since $\umul$ is additive and $\umul(\Sigma_0)=\Sigma_0$, we also have that 
$\umul(\Sigma_0 \cup \post(X)) = \umul(\Sigma_0) \cup \umul(\post(X))= \Sigma_0 \cup \umul(\post(X))$, and this allows us to conclude. 
\qed
\end{proof}

Consequently, $\lfp(\lambda X. \Sigma_0 \cup \umul(\post(X)))\subseteq P \Lra 
\lfp(\lambda X.\Sigma_0 \cup \post(\umul(X)) \cup \umul(X))\subseteq P$ holds. Since $\lambda X.\post(\umul(X)) \cup \umul(X)$ is additive, we can apply the co-inductive invariant principle 
\eqref{dinv} by considering its adjoint function, which is as follows: 
\begin{align*}
\post (\umul(X)) \cup \umul(X) \subseteq Y
&\Lra 
\post (\umul(X)) \subseteq Y \wedge \umul(X) \subseteq Y\Lra \\
X \subseteq \lmul(\pret(Y)) \wedge X \subseteq \lmul(Y) & \Lra
X \subseteq \lmul(\pret(Y)) \cap \lmul(Y) 
\end{align*}
Thus, by Lemma~\ref{lemma9} and \eqref{dualization}, we obtain:
\begin{align*}
\lfp(\lambda X. \umul(\Sigma_0) \cup \umul(\post(X))) \subseteq P &
\Lra 
\umul(\Sigma_0) \subseteq \gfp(\lambda X. \lmul(\pret(X) \cap X \cap P))\\
&\Lra 
\Sigma_0 \subseteq \gfp(\lambda X. \lmul(\pret(X) \cap X \cap P))
\end{align*}
and, in turn, 
by the abstract inductive invariant principle (Lemma~\ref{char}~(a) for ucos)
applied to $\umul \circ (\lambda X.\Sigma_0 \cup \post(X)) =
\lambda X. \umul(\Sigma_0) \cup \umul(\post(X))$ we get:
\begin{equation*}
\exists \phi \in L. \Sigma_0 \subseteq \phi\wedge \post(\phi)\subseteq \phi \wedge \phi \subseteq P \Lra 
\Sigma_0 \subseteq \gfp(\lambda X. \lmul(\pret(X) \cap X \cap P)).
\end{equation*}
This leads us to use the algorithm introduced by 
Cousot~\cite[Algorithm~2]{cou00} 
which synthesizes 
an inductive invariant by using a co-inductive method
that applies Knaster-Tarski theorem  
to compute the
iterates of the greatest fixpoint of 
$\lambda X. \lmul(\pret(X)$ $\cap X \cap P)$ as long as the current iterate $I_1$ contains the initial states in $\Sigma_0$:

\medskip
{\small
\begin{algorithm}[H]
\caption{Co-inductive backward abstract inductive invariant synthesis.}\label{algo1}

$I_1:=\Sigma$\; 
\While(\tcp*[h]{Loop invariant: $I_1\in L$}){$\Sigma_0 \subseteq I_1$}{ 
\lIf{$(I_1 =\lmul(\pret(I_1)\cap I_1 \cap P))$}{\Return{$I_1$  is an inductive {invariant in $L$}}}
$I_1:=I_1\cap \lmul(\pret(I_1)\cap I_1 \cap P)$\;
}
\Return{\textit{no inductive invariant in $L$}}\;
\end{algorithm}
}

\medskip
\noindent
Since, $\pret$ is computable and, by Lemma~\ref{lemma6}, 
$\tuple{\lmul,\subseteq}=\tuple{L,\subseteq}$ is a wqo, 
we immediately obtain that Algorithm~\ref{algo1} is correct
and terminating. Furthermore, if Algorithm~\ref{algo1} 
outputs an inductive 
invariant $I_1$ proving the property $P$ then $I_1$ is the \emph{greatest}
inductive invariant proving $P$.  
It turns out that Algorithm~\ref{algo1} exactly coincides 
with the synthesis algorithm by Padon et al.~\cite{padon16}, 
which is replicated below as Algorithm~{\rm\ref{algo3}}.

\medskip
{\small
\begin{algorithm}[H]
\caption{Inductive invariant algorithm by \cite{padon16}.}\label{algo3}

$I_2 := \Sigma$\;

\While(\tcp*[h]{Loop invariant: $I_2\in L$}){$I_2$ \textit{is not an inductive invariant}}{
\lIf{$\Sigma_0\not\subseteq I_2$}{\Return{\textit{no inductive invariant in $L$}}}
\textbf{choose} $s\in \Sigma$ \textit{as a counterexample to inductiveness of} $I_2$\;
	$I_2:=I_2\cap \av_L(s)$\;
}
\Return{$I_2$ is an inductive invariant in $L$}\;
\end{algorithm}
}

\begin{lemma}\label{lemma10}
Let $I\in L$ and $\Sigma_0\subseteq I$.
\begin{enumerate}[{\rm (a)}]
\item 
there exists  a counterexample to inductiveness of $I$
iff 
 $I\not \subseteq \pret(I)\cap P$ iff $I \neq \lmul(\pret(I)\cap I \cap P)$.
\item 
If $s\in \Sigma$ is a counterexample to inductiveness of $I$ then
$\lmul(\pret(I)\cap I \cap P) \subseteq \av_L(s)$. 
\end{enumerate}
\end{lemma}
\begin{proof}

\noindent (a) Under the assumption that $\Sigma_0\subseteq I$, 
$s$ is a counterexample to inductiveness of $I$ iff 
$(s\in I \wedge \post(s)\not\subseteq I) \vee s\in I\cap \neg P$.
Observe that $\post(s)\not\subseteq I$ 
iff $s\not\in \pret(I)$, so that $\exists s\in \Sigma. 
s\in I \wedge \post(s)\not\subseteq I$ iff $I\not\subseteq \pret(I)$. 
Hence, $\exists s\in \Sigma.(s\in I \wedge \post(s)\not\subseteq I) \vee s\in I\cap \neg P$ iff $I\not \subseteq \pret(I)\cap P$. 
Also:
\begin{align*}
I = \lmul(\pret(I)\cap I \cap P) &\Lra \quad\text{[as $I\in L$]}\\
I=\lmul(I) = \lmul(\pret(I)\cap I \cap P) &\Lra \quad\text{[as $\pret(I)\cap I \cap P\subseteq I$]}\\
I=\lmul(I) \subseteq \lmul(\pret(I)\cap I \cap P) &\Lra 
\quad\text{[as $\lmul$ is a lco]}\\
I=\lmul(I)\subseteq \pret(I)\cap I\cap P &\Lra\\
I\subseteq \pret(I)\cap P
\end{align*}

\noindent (b) The proof of point (a) shows  that
if $s\in \Sigma$ is a counterexample to inductiveness of $I$ 
then $s\in I$ and $s\not\in \pret(I) \cap P$. Then, $\pret(I) \cap P
\subseteq \neg\{s\}$, so that, by monotonicity of $\lmul$, 
$\lmul(\pret(I) \cap P)
\subseteq \lmul(\neg\{s\})$ and, in turn, $\lmul(\pret(I)\cap I \cap P)
\subseteq \lmul(\neg\{s\})= \av_L(s)$.
 \qed
\end{proof}

\begin{theorem}\label{theo-algo}
Algorithm~{\rm\ref{algo1}} $=$ Algorithm~{\rm\ref{algo3}}. 
\end{theorem}
\begin{proof} 
Consider the following variation of Algorithm~\ref{algo1}:

\medskip
\noindent
\begin{algorithm}[H]
\caption{A  modification of Algorithm~\ref{algo1}.}\label{algo2}

$I_4 := \Sigma$\;

\While(\tcp*[h]{Invariant: $I_4\in L$}){$\Sigma_0 \subseteq I_4$}{ 
\lIf{$(I_4\smallsetminus (\pret(I_4)\cap P)= \varnothing)$}{\Return{$I_4$  is an inductive invariant in $L$}}
\textbf{choose} $s\in I_4\smallsetminus (\pret(I_4)\cap P)$\;
$I_4:=I_4\cap \lmul(\{s\})$\;
}
\Return{\textit{no inductive invariant in $L$}}\;
\end{algorithm}

\medskip
\noindent
Algorithm~\ref{algo1} returns $I_1\in L$ iff 
$I_1=\gfp(\lambda X. \lmul(\pret(X) \cap X \cap P))$.
In this case, by Lemma~\ref{lemma10}~(a), since $\Sigma_0\subseteq I_1$ holds, $I_1$
is an (actually, the greatest) 
inductive invariant in $L$.  
 Otherwise, 
Algorithm~\ref{algo1} returns ``no inductive invariant in $L$''. 
By Lemma~\ref{lemma10}~(a), Algorithm~\ref{algo2} 
returns $I_4\in L$ iff $I_4\subseteq \pret(I_4)\cap P$ iff
$I_4 = \lmul(\pret(I_4)\cap I_4 \cap P)$, otherwise it  
returns ``no inductive invariant in $L$''. 
Algorithm~\ref{algo3} 
returns $I_2\in L$ iff $I_2$ is an inductive invariant, otherwise
it  
returns ``no inductive invariant in $L$''. 
Let $I_k^n$ be the current candidate invariant of Algorithm~$k\in \{1,2,4\}$ at 
its $n$-th iteration and $I_k$ be the output invariant of  Algorithm~$k$.
By Lemma~\ref{lemma10}~(b), $I_1^n\subseteq I_4^n=I_2^n$, so that
$I_1\subseteq I_4=I_2$. Since $I_k$ are
fixpoints of $\lambda X.\lmul(\pret(X)\cap X \cap P)$ and $I_1$ is the greatest
fixpoint, it turns out that $I_1= I_4=I_2$.
\qed
\end{proof}

This shows that Algorithm~{\rm\ref{algo3}} in \cite{padon16} 
for a disjunctive GC-based abstract domain $A$
amounts to a backward static analysis (\textit{i.e.}, propagating $\pret$) 
using the best correct approximations in $A$ of the
transfer functions, as long as the ordering  
of $A$ 
guarantees its termination, \textit{e.g.}, because $A$ is well-founded. 

\begin{examplebf}
It turns out that $\Const$, as defined in Section~\ref{const-sec},
is a disjunctive abstract domain with finite height and its bca's 
of affine assignments and linear Boolean guards are computable. Thus,
$\Const$ satisfies the hypotheses guaranteeing the correctness and 
termination of Algorithm~\ref{algo1} (and \ref{algo3}). 
Let us apply Algorithm~\ref{algo1} to the program $\cP$ of Example~\ref{ex-one} 
in Fig.~\ref{fig2}, still with
$\Sigma_0=\{q_1\}\times \bZ^2$ such that 
$\alpha_{\Const}(\Sigma_0) = \ok{\big\langle\tuple{q_1,(\top,\top)},}$ 
$\ok{\tuple{q_2,(\bot,\bot)},\tuple{q_3,(\bot,\bot)},\tuple{q_4,(\bot,\bot)}\big\rangle}$
and for the valid property
$P^\sharp = \ok{\big\langle \tuple{q_1,(\top,\top)},}$
$\ok{\tuple{q_2,(\top,2)}, \! \tuple{q_3,(\top,\top)},}$
$\ok{\tuple{q_4,(\top,\top)}\big\rangle} \in \ok{\Const_2^{|Q|}}$.
It is easy to check that
Algorithm~\ref{algo1} computes the following sequence $I^{k+1}=\alpha_{\Const}(\pret_{\cP}(\gamma_{\Const}(I^k))) \sqcap_{\Const} I^k \sqcap_{\Const} P^\sharp$ in $\Const_2^{|Q|}$:
\begin{align*}
&I^0=\dot{\alpha}_{\Const}(\Sigma)=\big\langle\tuple{q_1,(\top,\top)}, 
\tuple{q_2,(\top,\top)},\tuple{q_3,(\top,\top)},\tuple{q_4,(\top,\top)}\big\rangle\\
&I^1=\big\langle\tuple{q_1,(\top,\top)}, 
\tuple{q_2,(\top,2)},\tuple{q_3,(\top,\top)},\tuple{q_4,(\top,\top)}\big\rangle\\
&I^2=\big\langle\tuple{q_1,(\top,\top)}, 
\tuple{q_2,(\top,2)},\tuple{q_3,(\top,1)},\tuple{q_4,(\top,\top)}\big\rangle=\\
&\quad =\dot{\alpha}_{\Const}(\pret_{\cP}(\dot{\gamma}_{\Const}(I^3))) \:\dot{\geq}_{\Const}\:
\dot{\alpha}_{\Const}(\Sigma_0)
\end{align*}
The computation of $I^2$ derives from
$\pret(\{q_2\}\times \bZ \times \{2\}\cup \{q_1,q_3,q_4\}\times \bZ^2)
=\{q_1,q_2,q_4\}\times \bZ^2 \cup \{q_3\}\times \bZ\times \{1\}$, so
$\gamma_{\Const} (I^1) \cap \gamma_{\Const}(\alpha_{\Const}(\{q_1,q_2,q_4\}\times \bZ^2 \cup \{q_3\}\times \bZ\times \{1\}))= \{q_1,q_4\}\times \bZ^2 \cup \{q_2\}\times \bZ\times \{2\} \cup \{q_3\}\times \bZ\times \{1\}$. 
Thus, $I^2$ is returned by  Alg.~\ref{algo1} and is the greatest inductive invariant in $\Const$ which proves $P^\sharp$.
In particular, the abstract inductive invariant $I^2$ is strictly weaker than
the inductive invariant $\ok{\big\langle\tuple{q_1,(\top,\top)}, 
\tuple{q_2,(\top,2)},\tuple{q_3,(\top,1)},\tuple{q_4,(\top,2)}\big\rangle}$ computed in Example~\ref{ex-one} for proving
the same property $P^\sharp$ by the algorithm $\AInv$ of Corollary~\ref{coro-algo}.
\qed
\end{examplebf}

\subsection{Backward and Forward Algorithms}
We have that Algorithm~\ref{algo1} is backward because it applies 
$\pret$, for termination it requires 
that the abstract domain $\tuple{\umul,\subseteq}$ is DCC 
and turns out to be
the dual of the forward algorithm $\AInv$ provided by 
Corollary~\ref{coro-algo} for $\post$ and 
requiring that
$\tuple{\umul,\subseteq}$ is ACC. 
This latter assumption would be satisfied
by dualizing the wqo $\sqsubseteq_L$, \textit{i.e.}, 
by requiring that $\tuple{\Sigma,\sqsupseteq_L}$ is a wqo. 

A different gfp-based \emph{forward} algorithm based can be  
designed by observing (as in \cite[Section~3]{CC99}) 
that $\lfp(\lambda X. \Sigma_0 \cup \post(X))\subseteq P$ iff
$\lfp(\lambda X.\neg P \cup \pre(X))\subseteq \neg \Sigma_0$. 
Here, the dualization provided by the equivalence \eqref{dualization} yields:
\begin{equation*}
\lfp(\lambda X. \neg P \cup \umul(\pre(X))) \subseteq \neg\Sigma_0 \Lra 
\neg P \subseteq \gfp(\lambda X. \lmul(\postt(X) \cap X \cap \neg \Sigma_0)).
\end{equation*}
This induces 
the following \emph{co-inductive forward algorithm} which relies on the state
transformer $\postt$ and is terminating when 
$\tuple{\umul,\subseteq}$ is assumed to be DCC:

\medskip
\begin{algorithm}[H]
\caption{Co-inductive forward abstract inductive invariant synthesis.}\label{algo-4}

$I:=\Sigma$\; 
\While(\tcp*[h]{Loop invariant: $I\in L$}){$\neg P \subseteq I$}{
\lIf{$(I =\lmul(\postt(I)\cap I \cap \neg\Sigma_0))$}{\Return{$I$  is an inductive {invariant in $L$}}}
$I:=I\cap \lmul(\postt(I)\cap I \cap \neg \Sigma_0)$\;
}
\Return{\textit{no inductive invariant in $L$}}\;
\end{algorithm}

\medskip
Furthermore, by dualizing the technique and the algorithm described in 
\cite[Section~4.3]{CC99} for $\post$ and $\pre$, one could also design
a more efficient combined forward/backward synthesis algorithm
which simultaneously make backward, by $\pret$, and forward, by $\postt$, 
steps.

\section{Future Work}
As hinted by Monniaux~\cite{mon19}, 
results of undecidability to the question~\eqref{ainv} for some abstract domain $A$ display a foundational trait since they ``vindicate'' (often years of intense) research on precise and efficient algorithms for \emph{approximate} 
static program analysis on $A$.  To the best of our knowledge, few undecidability results are available:
an undecidability result by Monniaux~\cite[Theorem~1]{mon19} for the abstraction of convex polyhedra~\cite{CH78} and by  Fijalkow et al.~\cite[Theorem~1]{worrell19} for semilinear sets, \textit{i.e.}\ 
finite unions of convex polyhedra. 
However, convex polyhedra and semilinear sets cannot be defined by a Galois connection and therefore do not satisfy  our Assumption~\ref{assu}. 
As an interesting and stimulating future work we plan to investigate whether 
the abstract inductive invariant principle could be exploited to provide a reduction of the undecidability of 
the question~\eqref{ainv} for abstract domains which satisfy 
Assumption~\ref{assu} and, in view of the characterization 
of fixpoint completeness given in Section~\ref{sec-char},
for transfer functions which are not fixpoint complete.  

We also plan to study whether complete abstractions 
can play a role in the decidability 
result by Hrushovski et al.~\cite{worrell18}
on the computation of the strongest polynomial
invariant of an affine program. This hard result in \cite{worrell18} relies on the Zariski closure, %
which is continuous for affine functions. 
The observation here is that the Zariski closure 
is pointwise complete for the transfer functions of affine programs, 
therefore fixpoint completeness 
for affine programs holds, and one could investigate whether the algorithm 
given 
in \cite{worrell18} may be viewed as a least fixpoint computation of a best correct approximation on the Zariski abstraction.


\end{document}